\documentclass[11pt,reqno]{amsart}

\usepackage[english]{babel}
\usepackage{amsmath}
\usepackage{amsfonts}
\usepackage{amssymb}
\usepackage{amsthm}
\usepackage{graphicx}
\usepackage{epstopdf}
\usepackage{enumitem}
\usepackage{geometry}
\usepackage{hyperref}
\usepackage{tikz}
\usepackage[all]{xy}
\usepackage{multirow}
\usepackage{anysize}
\usepackage{setspace}

\newtheorem{theorem}{Theorem}[section]
\newtheorem{definition}[theorem]{Definition}
\newtheorem{fact}[theorem]{Fact}

\newtheorem{example}[theorem]{Example}
\newtheorem{proposition}[theorem]{Proposition}
\newtheorem{remark}[theorem]{Remark}
\newtheorem{corollary}[theorem]{Corollary}
\newtheorem{lemma}[theorem]{Lemma}

\newcommand{\beq}{\begin{equation}}
\newcommand{\eeq}{\end{equation}}

\marginsize{3.2cm}{3.2cm}{3.2cm}{3.2cm}

\begin{document}

\title{A solution concept for games with altruism and cooperation}

\author{Valerio Capraro}
\thanks{Address: Department of Mathematics, University of Southampton, Southampton, SO17 1BJ, UK}
\thanks{Email: V.Capraro@soton.ac.uk}

\keywords{Normal form games, solution concept, Prisoner's dilemma, Traveler's dilemma, cooperation, cumulative prospect theory, altruism.}

\thanks{We thank Joe Halpern and Marco Scarsini for helpful discussions over all stages of this research. We thank all participants to the EconCS Seminar at UC Berkeley and all participants to the Agents Seminar in Southampton for numerous comments and, in particular, we thank Rafael Frongillo and Christos Papadimitriou for two stimulating questions that led to improve a non-negligible part of this paper. We thank Kevin Leyton-Brown, Maria Polukarov, and James R. Wright for reading the last version of this article and suggesting many improvements. A special thank goes to my girlfriend, Chiara Napoleoni, student in Philosophy, with whom I had very interesting conversations about altruism.}

\maketitle

\textbf{Update Sept. 9th, 2013.} Parts of this paper (concerning social dilemmas) have been published in \cite{Ca13} and \cite{CVPJ}. This Working Paper is an attempt to extend the theory developed in those published paper to all normal form games.

\begin{abstract}
Over the years, numerous experiments have been accumulated to show that cooperation is not casual and depends on the payoffs of the game. These findings suggest that humans have attitude to cooperation by nature and the same person may act more or less cooperatively depending on the particular payoffs. In other words, people do not act a priori as single agents, but they forecast how the game would be played if they formed coalitions and then they play according to their best forecast.

In this paper we formalize this idea and we define a new solution concept for one-shot normal form games.

We prove that this \emph{cooperative equilibrium} exists for all finite games and it explains a number of different experimental findings, such as (1) the rate of cooperation in the Prisoner's dilemma depends on the cost-benefit ratio; (2) the rate of cooperation in the Traveler's dilemma depends on the bonus/penalty; (3) the rate of cooperation in the Publig Goods game depends on the pro-capite marginal return and on the numbers of players; (4) the rate of cooperation in the Bertrand competition depends on the number of players; (5) players tend to be fair in the bargaining problem; (6) players tend to be fair in the Ultimatum game; (7) players tend to be altruist in the Dictator game; (8) offers in the Ultimatum game are larger than offers in the Dictator game.
\\

JEL Classification: C71, C72.

\end{abstract}

\tableofcontents

\section{Introduction}

Since its foundation  by Morgenstern and von Neumann \cite{Mo-vN44}, the major challenge of modern game theory has been to predict which actions a human player would adopt in a strategic situation. A first prediction was proposed in an earlier paper by J. von Neumann \cite{vN28} for two-person zero-sum games and then generalized to every finite game by J. Nash in \cite{Na50a}. Since then Nash equilibrium has certainly been the most notable and used solution concept in game theory. Nevertheless, over the last sixty years, it has been realized that it makes poor predictions of human play and, indeed, a large number of experiments have been conducted on games for which it drammatically fails to predict human behavior. 

There are many reasons behind this failure. On the one hand, when there are multiple equilibria, it is not clear which one we should expect is going to be played. A whole stream of literature, finalized to the selection of one equilibrium, arose from this point, including the definitions of evolutionarily stable strategy \cite{MS-Pr73}, perfect equilibrium \cite{Se75}, trembling hand perfect equilibrium \cite{Se75}, proper equilibrium \cite{My78}, sequential equilibrium \cite{Kr-Wi82}, limit logit equilibrium \cite{MK-Pa95}, and, very recently, settled equilibrium \cite{My-We12}. 

On the other hand, the criticism of Nash equilibrium is motivated by more serious problems: there are examples of games with a unique Nash equilibrium which is not played by human players. Typical examples of such a fastidious situation are the Prisoner's Dilemma \cite{Fl52}, the Traveler's Dilemma \cite{Ba94}, and, more generally, every social dilemma \cite{Ko88}. This point has motivated another stream of literature devoted to the explanation of such deviations from Nash equilibria. Part of this literature tries to explain such deviations assuming that players make mistakes in the computation of the expected value of a strategy and therefore, assuming that errors are identically distributed, a player may also play non-optimal strategies with a probability described by a Weibull distribution. This intuition led to the foundation of the so-called quantal response equilibrium theory by McKelvey and Palfrey \cite{MK-Pa95}. A variant of this theory, called quantal level-k theory and proposed  by Stahl and P. Wilson in \cite{St-Wi94}, was recently shown to perform better in the prediction of human behavior \cite{Wr-LB10}. In the same paper, Wright and Leyton-Brown have also shown that quantal level-k theory predicts human behavior significantly better than all other behavioral models that have been proposed in the last decade, as the level-k theory \cite{CG-Cr-Br01} and the cognitive hierarchy model \cite{Ca-Ho-Ch04}. However, an obvious criticism of quantal level-k theory is that it is not scale invariant, contradicting one of the axioms of expected utility theory of Morgenstern and von Neumann\cite{Mo-vN47}. A perhaps more fundamental criticism stems from the fact that quantal level-k theory only makes use of some parameters describing either the incidence of errors that a player can make computing the expected utility of a strategy 
 or the fact that humans can perform only a bounded number of iterations of strategic reasoning.  
These features first imply that quantal level-k theory is not predictive, in the sense that one has to conduct experiments to estimate the parameters; second, they imply that quantal level-k theory intrinsically affirms that deviation from Nash equilibria can descend only from two causes, computational mistakes and bounded rationality, that are hard to justify for games with very easy payoffs, like the Prisoner's Dilemma, or for games where the deviation from Nash equilibrium is particularly strong, like the Traveler's Dilemma with small bonus-penalty. 

Indeed, the general feeling is that the motivation must rely somewhere deeper and that Nash equilibrium should be replaced by a conceptually different solution concept that takes into account other features of human behavior and coincides with Nash equilibrium only in particular cases. The first studies in this direction
have been presented by Renou and Schlag \cite{Re-Sc09} and Halpern and Pass \cite{Ha-Pa12}, by Halpern and Rong \cite{Ha-Ro10}, by Halpern and Pass \cite{Ha-Pa11}, by Jamroga and Melissen \cite{Ja-Me11}, and by Adam and Ehud Kalai \cite{Ka-Ka13}. Nevertheless, even though these solution concepts can explain deviations from Nash equilibria in some particular games, all of them make unreasonable predictions for many games of interest. For instance, the \emph{maximum perfect cooperative equilibrium} introduced in \cite{Ha-Ro10} is too rigid and predicts cooperation for sure in the Prisoner's and Traveler's Dilemmas, contradicting the experimental data collected in \cite{Ca-Go-Go-Ho99}, \cite{Go-Ho01}, \cite{Be-Ca-Na05}, \cite{Ba-Be-St11}, \cite{HRZ11}, \cite{DEJR12}, \cite{Fu-Ra-Dr12}, \cite{RGN12}. The \emph{iterated regret minimization} procedure introduced in \cite{Re-Sc09} and \cite{Ha-Pa12} can explain deviations towards cooperation in some variants of the Traveler's Dilemma, the Bertrand competition, the Centipede Game, and other games of interest, but it does not predict deviation towards cooperation in the Prisoner's Dilemma \cite{HRZ11}, \cite{DEJR12}, \cite{Fu-Ra-Dr12}, \cite{RGN12} and in the public good game \cite{Le95}, it cannot explain altruistic behaviors in the ultimatum game \cite{Fe-Sc99} and in the dictator game \cite{En11}, and makes unreasonable predictions for the Traveler's dilemma with punishment (see Example \ref{ex:punish travelers}), and a certain zero-sum game (see Example \ref{ex:sure gain}). The solution concept defined using algorithmic rationability in \cite{Ha-Pa11} can explain deviation towards cooperation in the iterated Prisoner's and Traveler's dilemmas, but it does not predict deviation towards cooperation in one-shot versions of the Prisoner's dilemma or in one-shot versions of the Traveler's dilemma with very small bonus-penalty, contradicting the experimental data reported in \cite{Go-Ho01}, \cite{Be-Ca-Na05}, \cite{HRZ11}, \cite{DEJR12}, \cite{Fu-Ra-Dr12}, \cite{RGN12}. The \emph{farsighted pre-equilibrium} introduced in \cite{Ja-Me11} is too rigid. For instance, the Prisoner's dilemma has two farsighted pre-equilibria, which coincide with Rabin's fairness equilibria \cite{Ra93}, where both players either cooperate or defect for sure. This contradicts the experimental data reported in \cite{HRZ11}, \cite{DEJR12}, \cite{Fu-Ra-Dr12}, \cite{RGN12}, which suggest that humans tend to play a mixed strategy. Finally, the coco value introduced by Adam and Ehud Kalai in \cite{Ka-Ka13}, unifying and developing previous works by Nash \cite{Na53}, Raiffa \cite{Rai53}, and E.Kalai-Rosenthal \cite{Ka-Ro78}, also appears to be too rigid. For instance, if two agents played the Prisoner's dilemma according to the coco value, then they would both cooperate for sure. This prediction contradicts the experimental data collected in \cite{HRZ11}, \cite{DEJR12}, \cite{Fu-Ra-Dr12}, \cite{RGN12}.

In this paper we try to attribute the failure of all these attempts to two basic problems.

The first problem is the use of utility functions in the very definition of a game. Indeed, the experimental evidence have shown that expected utility theory fails to predict the behavior of decision makers \cite{Al53}, \cite{Ka-Tv00}, \cite{St00}.

This problem could be theoretically overcome replacing utility functions with gain functions and applying Kahneman-Tversky's cumulative prospect theory \cite{Tv-Ka92}. But one can easily convince himself that in most cases such a replacement could explain only \emph{quantitative deviations}.

The second problem is indeed that experiments conducted on the Prisoner's dilemma, the Traveler's dilemma, Dictator game, and other games, show \emph{qualitative deviations} from classical solution concepts. These qualitative deviations suggest that humans are altruistic and have attitude to cooperation. 

These observations motivate the definition of a new solution concept, able to take into account altruism and cooperation and using gain functions instead of utility functions. This paper represents a first endeavour in this direction. Indeed, here we consider only one-shot normal form games where the players are completely anonymous, that is, they do not know each other and they are not allowed to exchange information\footnote{We mention that anonimity is not really a necessary assumption: the effect of any sort of contact among the players would be a different evaluation of the so-called \emph{prior probability} $\tau$. The point is that at the moment it is not clear how this \emph{prior probability} should be re-evaluated.}. The aim of this paper is to define a new solution concept for this class of games. This solution concept will be called \emph{cooperative equilibrium}. Indeed, we will see that altruism plays only a marginal role and the main idea behind this new equilibrium notion is the formalization of the following \emph{principle of cooperation}: 
\begin{itemize}
\item[(C)] Players try to forecast how the game would be played if they formed coalitions and then they play according to their best forecast.  
\end{itemize}

The study of cooperation in games is not a new idea. Economists, biologists, psychologists, sociologists, and political scientists, have been studying cooperation in social dilemmas for forty years. These classical approaches explain tendency to cooperation dividing people in proself and prosocial types \cite{Li84}, \cite{LWVW86}, \cite{KMM86}, \cite{KCC86}, \cite{ML88}, or appealing to forms of external control \cite{Ol65}, \cite{Ha68}, \cite{Da80}, or to long-term strategies in iterated games\cite{Ax84}. But, over the years many experiments have been accumulated to show cooperation even in one-shot social dilemmas without external control \cite{Is-Wa88}, \cite{Co-DJ-Fo-Ro96}, \cite{Go-Ho01}, \cite{Be-Ca-Na05}, \cite{DRFN08}, \cite{HRZ11}, \cite{DEJR12}. These and other earlier experiments \cite{Ke-Kr72}, \cite{BSKM76}, \cite{KSK80}, \cite{IWT84} have also shown that the rate of cooperation in the same game depends on the particular payoffs, suggesting that most likely humans cannot be merely divided in proself and prosocial types, but they are engaged in some sort of indirect reciprocity \cite{No-Si98}, \cite{No06} and the same person may behave more or less cooperatively depending on the payoffs. In other words, humans have attitude to cooperation by nature.

To the best of our knowledge, this is the first attempt to lift this well known tendency to cooperate up to a general principle which is nothing more than a deeper and smarter realization of selfishness.

The idea to formalize the principle of cooperation and define the cooperative equilibrium can be briefly summarized as follows:
\begin{itemize}
\item We assume that players do not act a priori as single players, but they try to forecast how the game would be played if they formed coalitions.
\item Each \emph{forecast} is represented by a number $v_i(p)$, called \emph{value of the coalition structure $p$ for player $i$}, which is a measure of the expected gain of player $i$ when she plays according to the coalition structure $p$. 
\item The numbers $v_i(p)$ induce a sort of common beliefs: we consider the \emph{induced game} $\text{Ind}(\mathcal G,p)$ which differs from the original game $\mathcal G$ only for the set of allowed profiles of mixed strategies: the profiles of mixed strategies allowed in $\text{Ind}(\mathcal G,p)$ are the profiles $(\sigma_1,\ldots,\sigma_N)$ such that $u_i(\sigma_1,\ldots,\sigma_N)\geq v_i(p)$, for any player $i$. 
\item The \emph{exact} cooperative equilibrium is one where player $i$ plays an equilibrium of the game $\text{Ind}(\mathcal G,p)$ induced by a coalition structure which maximizes the value function $v_i$\footnote{The word \emph{exact} means that, since players can have bounded rationality or can make mistakes in the computations, one can also define a quantal cooperative equilibrium borrowing ideas from quantal response equilibrium and quantal level-k theory and say that player $i$ plays with probability $e^{\lambda v_i(p)}/\sum_pe^{\lambda v_i(p)}$ a quantal response equilibrium or a quantal level-k equilibrium of the game $\text{Ind}(\mathcal G,p)$.}.
\item The notion of equilibrium for the induced game $\text{Ind}(\mathcal G,p)$ is not defined using classical Nash equilibrium, but using a prospect theoretical analogue.
\end{itemize}

In order to apply prospect theory we must replace utility functions by gain functions, that are, functions whose values represent the monetary outcomes or, more generally, the quantity of some good which is won or lost by a player. This replacement comes at the price that we must take into account explicitly new data that were implicitly included in the utility functions. Indeed, while utility functions were supposed to contain all relevant information about players' preferences, gain functions do contain only the quantity of some good which is won or lost by the players. These new data include the fairness functions $f_i$ and the altruism functions $a_{ij}$.  An interesting feature of the cooperative equilibrium is that, in many games of interest, it does not depend on these functions. This implies that the cooperative equilibrium is a predictive solution concept for many games of interest. A bit more precisely, in this paper we prove the following statements.

\begin{fact}
{\rm The cooperative equilibrium for the Prisoner's dilemma is predictive (i.e., it does not depend on fairness functions and altruism functions) and has the following property: the predicted rate of cooperation increases as the cost-benefit ratio increases.}
\end{fact}

\begin{fact}
{\rm The cooperative equilibrium for the Traveler's dilemma is predictive and has the following property: the predicted rate of cooperation decreases as the bonus/penalty increases.}
\end{fact}

\begin{fact}
{\rm The cooperative equilibrium for the Bertrand competition is predictive and it has the following property: the predicted rate of cooperation decreases as the numbers of players increase.}
\end{fact}

\begin{fact}
{\rm The cooperative equilibrium fits Kahneman-Knetsch-Thaler's experiment related to the ultimatum game.}
\end{fact}

\begin{fact}
{\rm The cooperative equilibrium for the public good game is predictive and it has the following properties: (1) the predicted rate of cooperation increases as the marginal return increases, and (2) the predicted rate of cooperation decreases as the number of players increases and then increases again as the number of players gets sufficiently large.}
\end{fact}

\begin{fact}
{\rm The cooperative equilibrium predicts the (50,50) solution in the Bargaining problem under \emph{natural} assumptions on the fairness functions.}
\end{fact}

Roughly speaking, the \emph{natural} assumption is that the two players have the same perception of money. We believe that this assumption is natural, since it is predictable that a bargain between a very rich person and a very poor person can have a different solution.

\begin{fact}
{\rm The cooperative equilibrium explains the experimental data collected for the dictator game, via altruism.}
\end{fact}

This happens just because we define the altruism in terms of human behavior in the dictator game. To treat the dictator game as the quintessence of altruism is certainly not a new idea \cite{Ha-Kr00}, \cite{BEN11}, \cite{DFR11}.

\begin{fact}
{\rm The cooperative equilibrium explain the experimental data collected for the ultimatum game, via a combination of cooperation and altruism.}
\end{fact}

In particular, the observation that offers in the ultimatum game are larger then the offers in the dictator game is explained in terms of cooperation, which is generated by the fact that the responder has the power to reject proposer's offer.

Another case where the cooperative equilibrium is only descriptive is when the mistakes that players can make in the computations have a very strong influence on the result. A typical example is the following.

\begin{fact}
{\rm The \emph{quantal} cooperative equilibrium explains Goeree-Holt's experiment on the asymmetric matching pennies.}
\end{fact}

The structure of the paper is as follows. In Section \ref{se:explicit form}, we define the so-called games in explicit form (see Definition \ref{defin:new}), where the word \emph{explicit} really emphasize the fact that we have to take into account explicitly new data (altruism functions and fairness functions). In Section \ref{se:informal} we describe informally the idea through a simple example that allows to motivate all main definitions of the theory. In Section \ref{se:cooperative equilibrium eut} we define the cooperative equilibrium for games in explicit form under expected utility theory, that is, without using cumulative prospect theory, and without using the altruism functions (see Definition \ref{defin:exact cooperative equilibrium eut}). The reason of this choice is that in most cases cumulative prospect theory can change predictions only quantitatively and not qualitatively and that, in most cases, altruism functions do not play any active role. Indeed, we compute the cooperative equilibrium (under expected utility theory and without using the altruism functions) for the Prisoner's Dilemma (see Examples \ref{ex:prisoner} and \ref{ex:prisoner2}), Traveler's Dilemma (see Examples \ref{ex:traveler} and \ref{ex:traveler2}), Nash bargaining problem (see Example \ref{ex:bargaining} and \ref{ex:bargaining2}), Bertrand competition (see Example \ref{ex:bertrand}), public goods game (see Example \ref{ex:public good}), the ultimatum game (see Example \ref{ex:ultimatum}), and a specific game of particular interest since iterated regret minimization theory fails to predict human behavior, whereas the cooperative equilibrium does (see Example \ref{ex:punish travelers}). We make a comparison between the predictions of the cooperative equilibrium and the experimental data and we show that they are always close. In Section \ref{se:towards cpt} we discuss a few examples where the replacement of expected utility theory by cumulative prospect theory starts playing an active role (see Examples \ref{ex:public goods} and \ref{ex:generalized coordination game 2}). Here it starts the ideal second part of the paper, devoted to the definition of the cooperative equilibrium for games in explicit form, using cumulative prospect theory and taking into account altruism. Before doing that, we take a short section, namely Section \ref{se:cpt}, to give a brief introduction to cumulative prospect theory. The definition of the cooperative equilibrium under cumulative prospect theory and taking into account altruism takes Sections \ref{se:playable} and \ref{se:cooperative equilibrium cpt}: in the former we define a procedure of iterated deletion of strategies using the altruism functions and we apply it to explain the experimental data collected for the dictator game (see Example \ref{ex:dictator}) and the ultimatum game (see Example \ref{ex:ultimatum2}); in the latter we repeat the construction done in Section \ref{se:cooperative equilibrium eut}, this time under cumulative prospect theory instead of expected utility theory. Theorem \ref{th:existence} shows that all finite games have a cooperative equilibrium. Part of Section \ref{se:playable} may be of intrinsic interest, since it contains the definition of super-dominated strategies\footnote{Joseph Halpern communicated to the author that he and Rafael Pass have independently introduced super-dominated strategies (under the name \emph{minimax dominated strategies}) in \cite{Ha-Pa13}.} (see Definition \ref{defin:sure gain}) and their application to solve a problem left open in \cite{Ha-Pa12} (see Example \ref{ex:more iterations}). Section \ref{se:conclusions} states a few important problems that should be addressed in future researches.\\

\section{Utility functions vs Gain functions: games in explicit form}\label{se:explicit form}

As mentioned in the Introduction, a major innovation that we propose is the use of \emph{gain functions} instead of utility functions. In this section we first elaborate on the reasons behind this choice and then we investigate the theoretical consequences of such a choice. First recall the classical definition of a game in normal form.

\begin{definition}\label{defin:game general}
{\rm A finite game in strategic or normal form is given by the following data:
\begin{itemize}
\item a finite set of players $P=\{1,2,\ldots,N\}$;
\item for each player $i\in P$, a finite set of strategies $S_i$;
\item for each player $i\in P$, a preference relation $\leq_i$ on $S:=S_1\times\ldots\times S_N$.
\end{itemize} }
\end{definition}

It is frequently convenient (and very often included in the definition of a game) to specify the players' preferences by giving real-valued \emph{utility functions} $u_i:S\to\mathbb R$ that represent them. The definition and the use of utility functions relies in Morgenstern and von Neumann's expected utility theory \cite{Mo-vN47}, where, to avoid problems such as risk aversion, they assumed that players' utility functions contain all relevant information about the players' preferences over strategy profiles. In this way, Nash was then able to formalize Bernoulli's principle that each player attempts to maximize her expected utility \cite{Be738} given that the other players attempt to do the same. The use of utility functions can certainly make the theory much easier, but it is problematic, since it has been observed that humans constantly violate the principles of expected utility theory. The very first of such examples was found by M. Allais in \cite{Al53} and many others are known nowadays (see, for instance, \cite{Ka-Tv00} and \cite{St00} for a large set of examples). For the sake of completeness, we briefly describe one of these experiments (see \cite{Ka-Tv79}, Problems 3 and 4). In this experiment 95 persons were asked to choose between:\\ 

\begin{itemize}
\item[L1.] A lottery where there is a probability of 0.80 to win 4000 and 0.20 to win nothing,
\item[L2.] A certain gain of 3000.\\
\end{itemize}

An expected utility maximizer would choose the lottery L1. However, Kahneman and Tversky reported that 80 per cent of the individuals chose the certain gain. The same 95 persons were then asked to choose between:\\ 

\begin{itemize}
\item[L1'.] A lottery with a 0.20 chance of winning 4000 and 0.80 of winning nothing,
\item[L2'.] A lottery with a 0.25 chance of winning 3000 and 0.75 of winning nothing. \\
\end{itemize} 

This time 65 per cent of the subjects chose the lottery L1', which is also the lottery maximizing expected utility. These two results contradict the so-called \emph{substitution axiom} in expected utility theory and show how people can behave as expected utility maximizers or not depending on the particular situation they are facing. 

An even more dramatic observation is that the evidence suggests that decision makers weight probabilities in a non-linear manner, whereas expected utility theory postulates that they weight probabilities linearly. Consider, for instance, the following example from \cite{Ka-Tv79}, p.283. Suppose that one is compelled to play Russian roulette. One would be willing to
pay much more to reduce the number of bullets from one to zero than from four to three.
However, in each case, the reduction in probability of a bullet ring is 1/6 and so, under
expected utility theory, the decision maker should be willing to pay the same amount. One possible explanation is that decision makers do not weight probabilities in a linear manner as postulated by expected utility theory.

These problems have been now overcome in decision theory thanks to the celebrated prospect theory \cite{Ka-Tv79} and cumulative prospect theory \cite{Tv-Ka92}. One of the very basic principles of (cumulative) prospect theory is that decision makers think in terms of gains and losses rather than in terms of their net assets; in other words, they think in term of gain functions rather than in terms of utility functions. This forces us to replace utility functions by gain functions. This replacement comes at a price: while utility functions were supposed to contain all relevant information about the players' preferences, gain functions do not contain such information. They must be taken into account separately. As we will remind in Section \ref{se:cpt}, risk aversion is taken into account by cumulative prospect theory. Among the remaining relevant information there are (at least) two deserving particular attention:
\begin{itemize}
\item[] \textbf{Altruism.} A player may prefer to renounce to part of her gain in order to favor another player.
\item[] \textbf{Perception of gains.} Two different players may have different perceptions about the same amount of gain.
\end{itemize}

To define formally a game in terms of gain functions, we introduce a unit of measurement $\mathfrak g$ (tipically one dollar, one euro ...) and postulate that to every action profile $s\in S$ and to every player $i\in P$ is associated a quantity $g_i(s)$ of $\mathfrak g$ which is lost or won by player $i$ when the strategy profile $s$ is played. We assume that the unit of measurement (e.g., the currency) is common to all players. The losses are expressed by negative integers and the wins by positive integers, so that $g_i(s)=2$ will mean, for instance, that, if the strategy profile $s$ is played, then player $i$ wins two units of the good $\mathfrak{g}$; analogously, $g_i(s)=-3$ will mean that, if the strategy profile $s$ is played, then player $i$ loses three units of the good $\mathfrak{g}$. 

Using the unit of measurement, we can take into account altruism and perception of gains as follows. \\ 

\textbf{Definition of the altruism functions.} We define a notion of altruism operationally, that is, the altruism functions can be theoretically computed running a pre-experiment. Consider a \emph{general dictator game} as follows. A proposer has an endowment of $y\in\mathbb N$ units of $\mathfrak g$ and a responder has got already $z\in\mathbb Z$ units of $\mathfrak g$. Let $k>0$, the proposer chooses $x\in\{0,1,\ldots,y\}$, to transfer to the responder, who gets $\lfloor kx\rfloor$, that is, the largest integer smaller than or equal to $kx$. In other words, we define the two player game $\text{Dict}(k,y,z)$ where the strategy set of the first player is $S_1=\{0,1,\ldots,y\}$ and the strategy set of the second player contains only one strategy, that we call A. The gain functions are
$$
g_1(x,A)=y-x\qquad\text{and}\qquad
g_2(x,A)=z+\lfloor kx\rfloor.
$$

\begin{definition}\label{defin:altruism function}
{\rm The altruism function $a_{ij}$ is the function $a_{ij}:\mathbb R^+\times\mathbb N\times \mathbb Z\to\mathbb N$ such that $a_{ij}(k,y,z)$ would be the offer of player $i$ to player $j$ if $i$ were the proposer and $j$ were the responder in $\text{Dict}(k,y,z)$.}
\end{definition}

\textbf{Definition of the fairness functions.} To capture perception of money, we assume that to each player $i\in P$ is associated a function $f_i:\{(x,y)\in\mathbb R^2 : x\geq y\}\to[0,\infty)$ whose role is to quantify how much player $i$ disappreciates to renounce to a gain of $x$ and accept a gain of $y$. The following are then natural requirements:
\begin{itemize}
\item $f_i$ is continuous,
\item if $x>y$, then $f_i(x,y)>0$,
\item if $x=y$, then $f_i(x,y)=0$,
\item for any fixed $x>0$, the function $f_i(x,\cdot)$ is strictly decreasing and strictly convex for positive $y$'s and strictly concave for negative $y$'s,
\item for any fixed $y$, the function $f_i(\cdot,y)$ is strictly increasing and strictly concave for positive $x$'s and strictly convex for negative $x$'s. 
\end{itemize}

The last two properties formalize the well-known \emph{diminishing sensitivity principle} \cite{Ka-Tv79}: the same difference of gains (resp. losses) is perceived smaller if the gains (resp. losses) are higher. Indeed, one possible way to define the functions $f_i$ is to use Kahneman-Tversky's value function $v$ and set $f_i(x,y)=v(x)-v(y)$. The problem of this definition is that it does not take into account that different players may have different perception of the same amount of money (think of the perception of 100 dollars of a very rich person and a very poor person).\\

Therefore, we are led to study the following object.
\begin{definition}\label{defin:new}
{\rm A finite game in explicit form $\mathcal G=\mathcal G(P,S,\mathfrak g, g,a,f)$ is given by the following data:
\begin{itemize}
\item a finite set of players $P=\{1,2,\ldots,N\}$;
\item for each player $i\in P$, a finite set of strategies $S_{i}$;
\item a \emph{good} $\mathfrak g$, which plays the role of a unit of measurement;
\item for each player $i\in P$, a function $g_{i}:S_{1}\times\ldots\times S_{N}\to\mathbb Z$, called \emph{gain function};
\item for each pair of players $(i,j)$, $i\neq j$, an \emph{altruism function} $a_{ij}$;
\item For each player $i\in P$, a \emph{fairness function} $f_i:\{(x,y)\in\mathbb R^2 : x\geq y\}\to\mathbb R$ verifying the properties above.
\end{itemize} }
\end{definition}

The terminology \emph{explicit} puts in evidence the fact that we must take into account explicitly all parameters that are usually considered implicit in the definition of utility functions. We are not saying that there are only three such parameters (altruism functions, fairness functions, and risk aversion) and this is indeed the first of a long series of points of the theory deserving more attention in future researches. In particular, there is some evidence that \emph{badness parameters} can play an important role in some games. We shall elaborate on this in Section \ref{se:conclusions}.

The purpose of the paper is to define a solution concept for games in explicit form taking into account altruism and cooperation and using cumulative prospect theory instead of expected utility theory. Nevertheless, we will see that
\begin{itemize}
\item in most cases the use of cumulative prospect theory instead of expected utility theory can change predictions only quantitatively and not qualitatively;
\item in most cases the altruism functions do not play any active role, since there are no players having a strategy which give a certain disadvantage to other players.
\end{itemize}

Consequently, we prefer to introduce the cooperative equilibrium in two steps. In the first one we keep expected utility theory and we do not use the altruism functions. The aim of the first step is only to formalize the principle of cooperation. We show that already this cooperative equilibrium under expected utility theory and without altruism can explain experimental data satisfactorily well. In Section \ref{se:towards cpt} we discuss some examples where the cooperative equilibrium under expected utility theory does not perform well because of the use of expected utility theory and because we did not take into account altruism and so we move towards the definition of the cooperative equilibrium under cumulative prospect theory and taking into account altruism.

\section{An informal sketch of the definition}\label{se:informal}

In this section we describe the cooperative equilibrium (under expected utility theory and without taking into account altruism) starting from an example. The idea is indeed very simple, even though the complete formalization requires a number of preliminary definitions that will be given in the next section.

Consider the following variant of the Traveler's dilemma. Two players have the same strategy set $S_1=S_2=\{180,181,\ldots,300\}$. The gain functions are

$$
g_1(x,y)=\left\{
  \begin{array}{lll}
    x+5, & \hbox{if $x<y$} \\
    x, & \hbox{if $x=y$}\\
    y-5, & \hbox{if $x>y$,}\\
  \end{array}
\right. \qquad\text{and}\qquad
g_2(x,y)=\left\{
  \begin{array}{lll}
    y+5, & \hbox{if $x>y$} \\
    y, & \hbox{if $x=y$}\\
    x-5, & \hbox{if $x<y$.}\\
  \end{array}
\right.
$$

The usual backward induction implies that $(180,180)$ is the unique Nash equilibrium. Nevertheless, numerous experimental studies reject this prediction and show that humans play significantly larger strategies. 

In the cooperative equilibrium, we formalize the idea that players forecast how the game would be played if they formed coalitions and then they play according to their best forecast.

Let us try to describe how this idea will be formalized. In a two-player game, as the Traveler's dilemma, there are only two possible coalition structures, the selfish coalition structure $p_s=(\{1\},\{2\})$ and the cooperative coalition structure $p_c=(\{1,2\})$. Let us analyze them:
\begin{itemize}
\item If agents play according to the selfish coalition structure, then \emph{by definition} they do not have any incentive to cooperate and therefore they would play the Nash equilibrium $(180,180)$. A Nash equilibrium is, by definition, stable, in the sense that no players have any incentives to change strategy. Consequently, both players would get $180$ for sure. In this case we say that the value of the selfish coalition structure is $180$ and we write $v(p_s)=180$.
\item Now, let us analyze the cooperative coalition structure $p_c$. The largest gain for each of the two agents, if they play together, is to get $300$, that is attained by the profile of strategies $(300,300)$. Nevertheless, each player knows that the other player may defect and play a smaller strategy and so the value of the cooperative coalition is not $300$, but we have to take into account possible deviations. Let us look at the problem from the point of view of player 1. The other player, player 2, may deviate and play the strategy 299 or the strategy 298, or the strategy 297, or the strategy 296, or the strategy 295 (indeed, all these strategies give at least the same gain as the strategy 300, if the first player is believed to play the strategy 300). In this case, the best that player 2 can obtain is $304$ (if she plays 299 and the first player plays 300) and so we say that the incentive to deviate from the coalition is $304-300=4$. We denote this number by $D_2(p_c)$.  Now, if player 2 decides to deviate from the coalition, she or he incurs in a \emph{risk} due to the fact that also player 1 can deviate from the coalition either to follow selfish interest or because player 1 is clever enough to understand that player 2 can deviate from the coalition and then player 1 decides to anticipate this move. The maximal risk that player 2 incurs trying to achieve her maximal gain is then attained when player 2 deviates to 299 and player 1 anticipates this deviation and play 298. In this case, player 2 would gain $g_2(298,299)=293$. So we say that the \emph{risk in deviating from the coalition structure $p_c$} is $R_2(p_c)=300-293=7$. We now interpret the number
$$
\tau_{1,\{2\}}(p_c)=\frac{D_2(p_c)}{D_2(p_c)+R_2(p_c)}=\frac{4}{11},
$$
as a sort of \emph{prior probability} that player 1 assigns to the event ``\emph{player 2 abandons the coalition structure $p_c$}''.  Consequently, we obtain also a number
$$
\tau_{1,\emptyset}(p_c)=1-\tau_{1,\{2\}}(p_c),
$$
which is interpreted as a \emph{prior probability} that player 1 assigns to the event ``\emph{nobody abandons the coalition structure $p_c$}''.

This probability measure will be now used to weight the numbers $e_{1,\emptyset}(p_c)$, representing the infimum of gains that player $1$ receives if nobody abandons the coalition, and $e_{1,\{2\}}(p_c)$, representing the infimum of gains that player $1$ receives if the second player abandons the coalition. Therefore, one has
$$
e_{1,\emptyset}(p_c)=300\qquad\text{and}\qquad e_{1,\{2\}}(p_c)=290,
$$
where the second number comes from the fact that the worst that can happen for player $1$ if the second player abandons the coalition and the first players does not abandon the coalition is in correspondence of the profile of strategies $(300,295)$ which gives a gain $290$ to the first player. Taking the average we obtain the value of the cooperative coalition for player 1
$$
v_1(p_c)=300\cdot\frac{7}{11}+290\cdot\frac{4}{11}\sim296.35.
$$
\end{itemize}

By symmetry one has $v_2(p_s)=v_1(p_s)=:v(p_s)=180$ and $v_2(p_c)=v_1(p_c)=:v(p_c)=296.35$. So one has $v(p_s)<v(p_c)$ and then the cooperative equilibrium predicts that the agents \emph{play according} to the cooperative coalition structure, since it gives a better forecast. The meaning of the word \emph{play according to $p_c$} has to be clarified. Indeed, since the profile $(300,300)$ is not stable, we cannot expect that the players play for sure the strategy 300. What we do is to interpret the values $v_i(p_c)$ as a sort of common beliefs: players simply keep only the profiles of strategies $\sigma=(\sigma_1,\sigma_2)$ such that $g_1(\sigma)\geq v_1(p_c)$ and $g_2(\sigma)\geq v_2(p_c)$. Computing the Nash equilibrium in this \emph{induced game} will give the cooperative equilibrium of the game that, in this case, is a mixed strategy which is \emph{supported between} 296 and 297. Observe that this is very close to the experimental data. Indeed, the one-shot version of this game was experimented by Goeree and Holt who reported that 80 per cent of subjects played a strategy between 290 and 300 with an average of 295 (see \cite{Go-Ho01}).\\

The purpose of the next section is to formalize the idea that we have just described. Indeed, even though the idea is very simple and in many relevant cases computations can be easily performed by hand (cf. Section \ref{se:examples}), the correct formalization requires the whole section \ref{se:cooperative equilibrium eut} because of the following technical problems:
\begin{itemize}
\item In the particular example that we have just described, the cooperative coalition structure leads to a one-player game with a unique Nash equilibrium, which is $(300,300)$. In general this will not happen and we should take into account that one Nash equilibrium can be less fair than another. For instance, the cooperative coalition structure in Nash bargaining problem leads to a one-player game with many Nash equilibria, but intuitively only the (50,50) solution is fair. 
\item The definition of deviation and risk is intuitively very simple, but the general mathematical formalization is not straightforward.
\end{itemize}

\section{The cooperative equilibrium under expected utility theory}\label{se:cooperative equilibrium eut}

Let $\mathcal G=\mathcal G(P,S,\mathfrak g, g,a,f)$ be a finite\footnote{It is well known that the study of infinite games can be very subtle. For instance, there is large consensus that, at least when the strategy sets do not have a natural structure of a standard Borel space, one must allow also purely finitely additive probability measures as mixed strategies, leading to the problem that even the mixed extension of the utility functions is not uniquely defined \cite{Ma97},\cite{St05},\cite{Ca-Mo12},\cite{Ca-Sc12}. In this first stage of the research we want to avoid all these technical issues and we focuse our attention only to finite games.} game in explicit form. As usual, to make notation lighter, we denote $S_{-i}$ the cartesian product of all the $S_j$'s but $S_i$. Let $\mathcal P(X)$ be the set of probability measures on the finite set $X$. If $\sigma=(\sigma_1,\ldots,\sigma_N)\in\mathcal P(S_1)\times\ldots\times\mathcal P(S_N)$, we denote by $\sigma_{-i}$ the $(N-1)$-dimensional vector of measures $(\sigma_1,\ldots,\sigma_{i-1},\sigma_{i+1},\ldots,\sigma_N)$ and, as usual in expected utility theory, we set 
$$
g_j(\sigma_i,\sigma_{-i})=g_j(\sigma):=\sum_{(s_1,\ldots,s_N)\in S}g_j(s_1,\ldots,s_N)\sigma_1(s_1)\cdot\ldots\cdot\sigma_N(s_N).
$$

Conversely, if $\sigma_i\in\mathcal P(S_i)$, for all $i\in P$, the notation $g_j(\sigma_i,\sigma_{-i})$ simply stands for the number $g_j(\sigma_1,\ldots,\sigma_N)$. 



The main idea behind our definition is the principle of cooperation, that is, players try to forecast how the game would be played if they formed coalitions and then they play according to their best forecast. Borrowing a well known terminology from the literature on coalition formation (cf. \cite{Ra08}), we give the following definition.

\begin{definition}\label{defin:coalition structure}
{\rm A coalition structure is a partition $p=(p_1,\ldots,p_k)$ of the player set $P$; that is, the $p_\alpha$'s are subsets of $P$ such that $p_\alpha\cap p_\beta=\emptyset$, for all $\alpha\neq\beta$, and $\bigcup p_\alpha=P$.}
\end{definition}

As mentioned in the Introduction, the idea is that each player $i\in P$ assigns a value to each coalition structure $p$ and then plays according to the coalition structure with highest value. As described in Section \ref{se:informal}, the idea to define the value of a coalition structure $p$ for player $i$ is to take an average of the following kind. Suppose that for all $J\subseteq P\setminus\{i\}$ we have defined a number $\tau_{i,J}(p)$ describing the probability that players in $J$ abandon the coalition structure $p$ and a number $e_{i,J}(p)$ describing the infimum of possible gains of player $i$ when players in $J$ abandon the coalition structure $p$. Then we (would) define
\begin{align}\label{eq:value}
v_i(p)=\sum_{J\subseteq P\setminus\{i\}}e_{i,J}(p)\tau_{i,J}(p).
\end{align}

Our aim is to give a reasonable definition for the numbers $e_{i,J}(p)$ and $\tau_{i,J}(p)$ under the assumption that players do not know each other and are not allowed to exchange information. Of course, this is only a real restriction of the theory: if the players know each other and/or are allowed to exchange information, this will reflect on the computation of the probability $\tau_{i,J}(p)$. 

Before defining the numbers $e_{i,J}(p)$ and $\tau_{i,J}(p)$, we need to understand what kind of strategies agree with the coalition structure $p$. Indeed, as mentioned in Section \ref{se:informal}, if $p\neq(\{1\},\ldots,\{N\})$ is not the selfish coalition structure, some profiles of strategies might not be \emph{acceptable} by the players in the same coalition because they do not share the gain in a fair way among the players belonging to the same coalition $p_\alpha$. We can define a notion of fairness making use of the fairness functions $f_i$. First observe that the hypothesis of working with gain functions expressed using the same unit of measurement for all players allows us to sum the gains of different players and, consequently, we can say that a coalition structure $p=(p_1,\ldots,p_k)$ generates a game with $k$ players as follows. The players are the sets $p_\alpha$ in the partition, the pure strategy set of $p_\alpha$ is $\prod_{i\in p_\alpha}S_i$, and the gain function of player $p_\alpha$ is 
\begin{align}\label{eq:grouping}
g_{p_\alpha}(s_1,\ldots,s_N)=\sum_{i\in p_\alpha}g_i(s_1,\ldots,s_N)
\end{align}

This game, that we denote by $\mathcal G_p$, has a non-empty set of Nash equilibria\footnote{If $p=(P)$ is the grand coalition, then $\mathcal G_p$ is a one-player game, whose Nash equilibria are all probability measures supported on the set of strategies maximizing the gain function.} that we denote by $\text{Nash}(\mathcal G_p)$. Since the players in the same $p_\alpha$ are ideally cooperating, not all Nash equilibria are acceptable, but only the ones that distribute the gain of the coalition $p_\alpha$ as fairly as possible among the players belonging to $p_\alpha$. 

To define the subset of \emph{fair} of \emph{acceptable} equilibria, fix $i\in P$ and consider the restricted function $\overline g_i=g_i|_{\text{Nash}(\mathcal G_p)}:\text{Nash}(\mathcal G_p)\to\mathbb R$. Since $\text{Nash}(\mathcal G_p)$ is compact and $g_i$ is continuous, we can find $\overline\sigma_i\in\text{Nash}(\mathcal G_p)$ maximizing $\overline g_i$.

\begin{definition}\label{defin:agreement}
{\rm The disagreement in playing the profile of strategy $\sigma\in\text{Nash}(\mathcal G_p)$ for the coalition $p_\alpha$ is the number
$$
\text{Dis}_{p_\alpha}(\sigma)=\sum_{i\in p_\alpha}f_i(g_i(\overline \sigma_i),g_i(\sigma))
$$ }
\end{definition}

Recalling that the number $f_i(x,y)$ represents how much player $i$ disappreciates to renounce to a gain of $x$ and accept a gain of $y\leq x$, we obtain that, in to order to have a fair distribution of the gain among the players in the coalition $p_\alpha$, the disagreement $\text{Dis}_{p_\alpha}$ must be minimized. 

\begin{definition}\label{defin:cooperative equilibrium}
{\rm The Nash equilibrium $\sigma\in\text{Nash}(\mathcal G_p)$ is \emph{acceptable} or \emph{fair} for the coalition $p_\alpha$, if $\sigma$ minimizes $\text{Dis}_{p_\alpha}(\sigma)$.
}
\end{definition}

Since the set of Nash equilibria of a finite game is compact and since the functions $f_i$ are continuous, it follows that the set of acceptable equilibria is non-empty and compact.

Let us say explicitly that this is the unique point where we use the functions $f_i$. It follows, that, for a game $\mathcal G$ such that every game $\mathcal G_p$ has a unique Nash equilibrium, the cooperative equilibrium does not depend on the functions $f_i$.

The importance of the hypotheses about strict convexity in the second variable and strict concavity in the first variable of the functions $f_i$ should be now clear and is however described in the first of the following series of examples.

\begin{example}\label{ex:bargaining}
{\rm Consider a finite version of Nash's bargaining problem \cite{Na50b} where two persons have the same strategy set $S_1=S_2=S=\{0,1,\ldots,100\}$ and the gain functions are as follows:

$$
g_1(x,y)=\left\{
  \begin{array}{lll}
    x, & \hbox{if $x+y\leq100$} \\
    0, & \hbox{if $x+y>100$,}\\
  \end{array}
\right. \qquad\text{and}\qquad
g_2(x,y)=\left\{
  \begin{array}{lll}
    y, & \hbox{if $x+y\leq100$} \\
    0, & \hbox{if $x+y>100.$}\\
  \end{array}
\right.
$$
As well known, this game has attracted attention from game theorists since, despite having many pure Nash equilibria, only one is intuitively \emph{natural}. Indeed, many papers have been devoted to select this natural equilibrium adding axioms (see \cite{Na50b}, \cite{Ka-Sm75}, and \cite{Ka77}) or using different solution concepts (see \cite{Ha-Ro10} and \cite{Ha-Pa12}).

Assume that the two players have the same perception of money, that is $f_1=f_2=:f$.
Consider the cooperative coalition $p_c=(\{1,2\})$ describing cooperation between the two players. The game $\mathcal G_{p_c}$ is a one-player game whose Nash equilibria are all pairs $(x,100-x)$, $x\in S_1$, and all probability measures on $S_1\times S_2$ supported on such pairs of strategies. Despite having all these Nash equilibria, the unique acceptable equilibrium for the game coalition is $(50,50)$. Indeed, one has
\begin{align*}
\text{Dis}_{p_c}\left(50,50\right)
&=f\left(100,50\right)+f\left(100,50\right)\\
&=f\left(100,\frac12\cdot100+\frac12\cdot0\right)+f\left(100,\frac12\cdot100+\frac12\cdot0\right)\\
&<\frac12f(100,100)+\frac12f(100,0)+\frac12f(100,100)+\frac12f(100,0)\\
&=f(100,100)+f(100,0)\\
&=\text{Dis}_{p_c}(100,0).
\end{align*}
Analogously, one gets $\text{Dis}_{p_c}\left(50,50\right)<\text{Dis}_{p_c}((x,100-x))$, for all $x\in\{0,1,\ldots,100\}$, $x\neq50$. Consequently, $\left(50,50\right)$ is the unique acceptable equilibrium for the cooperative coalition $p_c$.

Now let $p_s=(\{1\},\{2\})$ be the selfish coalition structure. Then the unique acceptable equilibrium for player 1 is $(100,0)$ and the unique acceptable Nash equilibria for player 2 is $(0,100)$. 
}
\end{example}

\begin{example}\label{ex:prisoner}
{\rm As second example, we consider the Prisoner's Dilemma. As well known, this famous game was originally introduced by Flood in \cite{Fl52}, where he reported on a series of experiments, one of which, now known as Prisoner's Dilemma, was conducted in 1950. Even though Flood's report is seriously questionable, as also observed by Nash himself (cf. \cite{Fl52}, pp. 24-25), it probably represents the first evidence that humans tend to cooperate in the Prisoner's Dilemma. This evidence has been confirmed in \cite{Co-DJ-Fo-Ro96}, where the authors observed a non-negligible percentage of cooperation even in one-shot version of the Prisoner's dilemma. 

Here we consider a parametrized version of the Prisoner's Dilemma, as follows. Two persons have the same strategy set $S_1=S_2=\{\text{C,D}\}$, where C stands for \emph{cooperate} and D stands for \emph{defect}. Let $\mu>0$, denote by $\mathcal G^{(\mu)}$ the game described by the following gains:
$$
  \begin{array}{ccc}
     & \text{C} & \text{D} \\
    \text{C} & 1+\mu,1+\mu & 0,2+\mu \\
    \text{D} & 2+\mu,0 & 1,1 \\
  \end{array}
$$
Therefore, the parameter $\mu$ plays the role of a reward for cooperating.
The intuition, motivated by similar experiments conducted on the Traveler's Dilemma (cf. Example \ref{ex:traveler}) or on the repeated Prisoner's dilemma \cite{DRFN08}, suggests that humans should play the selfish strategy $\text{D}$ for very small values of $\mu$ and tend to cooperate for very large values of $\mu$. This intuition is in fact so natural that Fudenberg, Rand, and Dreber, motivated by experimental results on the repeated Prisoner's dilemma, asked ``\emph{How do the strategies used vary with the gains to cooperation?}'' (cf. \cite{Fu-Ra-Dr12}, p.727, Question 4). We will propose an answer to this question (for one-shot Prisoner's dilemma) in Example \ref{ex:prisoner2}, where we will show that the cooperative equilibrium predicts a rate of cooperation depending on the particular gains and that such equilibrium is computable by a very simple formula (cf. Proposition \ref{prop:prisoner}). For now, let us just compute the acceptable Nash equilibria for the two partitions of $P=\{1,2\}$. Let $p_c=(\{1,2\})$ be the cooperative coalition structure, describing cooperation between the players. In this case we obtain a one-player game with gains:
$$
g_{p_c}(\text{C,C})=2+2\mu\qquad g_{p_c}(\text{C,D})=2+\mu\qquad g_{p_c}(\text{D,C})=2+\mu\qquad g_{p_c}(\text{D,D})=2.
$$
whose unique Nash equilibrium (i.e., the profile of strategies maximizing the payoff) is the \emph{cooperative} profile of strategies $(\text{C,C})$. Uniqueness implies that this equilibrium must be acceptable independently of the $f_i$'s. On the other hand, the selfish coalition structure $p_s=(\{1\},\{2\})$ generates the original game, whose unique equilibrium is, as well known, the \emph{defecting} profile of strategies $(\text{D,D})$. Also in this case, uniqueness implies that this equilibrium must be acceptable.
}
\end{example}

\begin{example}\label{ex:traveler}
{\rm Finally, we consider the Traveler's Dilemma. This game was introduced by Basu in \cite{Ba94} with the purpose to construct a game where Nash equilibrium makes unreasonable predictions. Basu's intuition was indeed confirmed by experiments on both one-shot and repeated treatments \cite{Ca-Go-Go-Ho99}, \cite{Go-Ho01}, \cite{Be-Ca-Na05}, \cite{Ba-Be-St11}. Fix a parameter $b\in\{2,3,\ldots,180\}$, two players have the same strategy set $S_1=S_2=\{180,181,\ldots,300\}$ and payoffs:
$$
g_1(x,y)=\left\{
  \begin{array}{lll}
    x+b, & \hbox{if $x<y$} \\
    x, & \hbox{if $x=y$}\\
    y-b, & \hbox{if $x>y$,}\\
  \end{array}
\right. \qquad\text{and}\qquad
g_2(x,y)=\left\{
  \begin{array}{lll}
    y+b, & \hbox{if $x>y$} \\
    y, & \hbox{if $x=y$}\\
    x-b, & \hbox{if $x<y$.}\\
  \end{array}
\right.
$$
This game has a unique Nash equilibrium, which is $(180,180)$. Nevertheless, it has been observed that humans tend to cooperate (i.e. play strategies close to $(300,300)$) for small values of $b$ and tend to be selfish (i.e., play strategies close to the Nash equilibrium $(180,180)$) for large values of $b$. This is indeed what the cooperative equilibrium predicts, as we will see in Example \ref{ex:traveler2}. For now, let us just compute the sets of acceptable equilibria for all partitions of $P=\{1,2\}$. Let $p_c=(\{1,2\})$ be the cooperative coalition structure, describing cooperation between the players. In this case we obtain a one-player game whose unique Nash equilibrium is attained by the \emph{cooperative} profile of strategies $(300,300)$. Uniqueness implies that this equilibrium must be acceptable. On the other hand, the selfish coalition structure $p_s=(\{1\},\{2\})$ gives rise to the unique Nash equilibrium of the game, which is $(180,180)$. Also in this case, uniqueness implies that this equilibrium must be acceptable.
}
\end{example}

Coming back to the description of the theory, we have gotten, for all partitions $p$ of the player set $P$ and for all sets $p_\alpha$ of the partition, a (compact) set of acceptable equilibria $\text{Acc}_{p_\alpha}(\mathcal G_p)$ for the coalition $p_\alpha$ inside the coalition structure $p$. Now we define the numbers $e_{i,J}(p)$ and $\tau_{J}(p)$.\\

\textbf{Definition of the numbers $\tau_{i,J}(p)$.} We recall that the number $\tau_{i,J}(p)$ represents the probability that players i assigns to the event ``\emph{players in $J$ abandon the coalition structure $p$}''. Consequently, it is enough to define the numbers $\tau_{i,J}(p)$ when $J=\{j\}$ contains only one element. The other numbers can be indeed reconstructed assuming that the events ``\emph{player $j$ deviates from $p$}'' and ``\emph{player $k$ deviates from $p$}'' are independent. This assumption is natural in this context where players are not allowed to exchange information. 

Therefore, fix $j\in P$, with $j\neq i$. The definition of $\tau_{i,j}(p)$ is intuitively very simple. It will be a ratio
$$
\tau_{i,j}(p)=\frac{D_j(p)}{D_j(p)+R_j(p)},
$$
where:
\begin{itemize}
\item the number $D_j(p)$ represents the incentive for player $j$ to abandon the coalition structure $p$, that is, the maximal gain that player $j$ can get leaving the coalition;
\item the number $R_j(p)$ represents the risk that player $j$ takes leaving the coalition structure $p$, that is, the maximal loss that player $j$ can incur trying to achieve her maximal gain, assuming that also other players can abandon the coalition either to follow selfish interests or to anticipate player $j$'s defection.
\end{itemize}

To make this intuition formal, first define
\begin{align}\label{eq:maximal}
\widetilde M(p_\alpha,p):=\left\{\sigma\in\text{Acc}_{p_\alpha}(\mathcal G_p) : g_{p_\alpha}(\sigma) \text{ is maximal}\right\}.
\end{align}

The idea is indeed that players in the same coalition try to achieve their maximal joint gain but, doing that, there might be some conflicts among coalitions. Therefore, we are interested to look at the strategy profiles that can be constructed putting together pieces of strategies in the various $\widetilde M(p_\alpha,p)$. 

To this end, let us fix a piece of notation. For a given player $j$, let $\pi_j:\mathcal P(S_1)\times\ldots\times\mathcal P(S_N)\to\mathcal P(S_j)$ be the canonical projection. We may \emph{reconstruct} an element $\sigma\in\mathcal P(S_1)\times\ldots\times\mathcal P(S_N)$, through its projections and we write formally $\sigma=\bigotimes_{j=1}^N\pi_j(\sigma)$. Set
\begin{align}
M(p_\alpha,p):=\left\{\bigotimes_{i\in p_\alpha}\pi_i(\sigma) : \sigma\in\widetilde M(p_\alpha,p)\right\},
\end{align}
and then
\begin{align}
M(p)=\bigotimes_{\alpha=1}^k M(p_\alpha,p).
\end{align}

In words, $M(p)$ is the set of strategy profiles that can be constructed putting together pieces of acceptable equilibria maximizing the joint gain of each coalition.

\begin{remark}\label{rem:barycenter}
{\rm It is worth mentioning that in many relevant cases all sets $\widetilde M(p_\alpha,p)$ contain only one element and the computations get very simple and \emph{unambiguous}. However, in some cases, as in the route choice game, this set may contain multiple and theoretically even infinite elements. From a mathematical point of view, this is not a problem, since we need only compactness of the sets $\widetilde M(p_\alpha,p)$ and these sets are indeed compact. However, in some cases there might be a natural way to restrict the sets $\widetilde M(p_\alpha,p)$, leading to a computationally lighter and intuitively more natural definition. For instance, in games with \emph{particular symmetries}, as the basic route choice game\footnote{There are $2N$ players, each of which has to decide the route to go to work between two \emph{equivalent} routes.}, players are tipically indifferent among all pure Nash equilibria maximizing their gains and, therefore, it is natural to restrict the set $\widetilde M(p_\alpha,p)$ and take only its barycenter, which is, in this case, the uniform measure\footnote{It was reported in \cite{RKDG09} that players tend to play uniformly in the basic route choice game.}. 
Theoretically, this construction may be extended to every game, since $\widetilde M(p_\alpha,p)$ is always compact and so it has a barycenter (see \cite{Ha-Va89}, Sec. 3.b). But we do not think that the assumption that players in the same $p_\alpha$ are indifferent among all the acceptable strategies which maximize their joint gain is very general and it would not probably make sense in very asymmetric games. How to restrict the sets $\widetilde M(p_\alpha,p)$ is another point of the theory that deserves particular attention in a future research.}
\end{remark}

\begin{definition}\label{defin:k-deviation}
{\rm Let $\sigma\in\mathcal P(S_1)\times\ldots\times\mathcal P(S_N)$ be a profile of mixed strategies and $\sigma_k'\in\mathcal P(S_k)$. We say that $\sigma_k'$ is a $k$-deviation from $\sigma$ if $g_k(\sigma'_k,\sigma_{-k})\geq g_k(\sigma)$. }
\end{definition}

Now we can finally move towards the definition of incentive and risk. We recall, that we have fixed a coalition structure $p$ and two players $i,j\in P$, with $i\neq j$ and we want to define the incentive and risk for player $j$ to abandon the coalition structure. Let $\text{Dev}_j(p)$ denote the set strategies of player $j$ that are $j$-deviation from at least one strategy in $M(p)$.

\begin{definition}
{\rm The incentive for player $j$ to deviate from the coalition structure $p$ is
\begin{align}
D_j(p):=\max\left\{g_j(\sigma_j',\sigma_{-j})-g_j(\sigma) : (\sigma,\sigma_j')\in\text{Dev}_j(p)\right\}.
\end{align} }
\end{definition}
Observe that $D_j(p)$ is attained since the set $\text{Dev}_j(p)$ is compact.

If $D_j(p)=0$, then $j$ does not gain anything by leaving the coalition and therefore $j$ does not have any incentives to abandon the coalition structure $p$. If it is the case, we simply define $\tau_{i,j}(p)=0$.

Consider now the more interesting case $D_j(p)>0$, where player $j$ has an actual incentive to deviate from the coalition structure $p$. If $j$ decides to leave $p$, it may happen that she loses part of her gain if other players decide to abandon $p$ either to follow selfish interests or to answer player $j$'s defection. To quantify this risk, we first introduce some notation. Let $(\sigma,\sigma_j')\in\text{Dev}_j(p)$ such that $D_j(p)$ is attained. Call $T(\sigma,\sigma_j')$ the set of $\sigma_{-j}'\in\bigotimes_{i\neq j}\mathcal P(S_i)$ such that
\begin{itemize}
\item  $g_j(\sigma)-g_j(\sigma_j',\sigma_{-j}')>0$,
\item there is $k\in P\setminus\{j\}$ such that $\pi_k(\sigma_{-j}')$ is a $k$-deviation from either $\sigma$ or $(\sigma_{-j},\sigma_j')$.
\end{itemize}
Thus we quantify the risk by
\begin{align}\label{eq:risk}
R_j(p):=\sup\left\{g_j(\sigma)-g_j(\sigma_j',\sigma_{-j}')\right\},
\end{align}
where the supremum is taken over all
\begin{enumerate}
\item[(A)] $(\sigma,\sigma_j')\in\text{Dev}_j(p)$ such that $D_j(p)$ is attained,
\item[(B)] $\sigma_{-j}'\in T(\sigma,\sigma_j')$.
\end{enumerate}
The requirement (A) is motivated by the fact that if player $j$ believes that she can leave the coalition structure $p$ to follow selfish interests, then she must take into account that also other players may deviate from $p$ either to follow selfish interests or because they are clever enough to anticipate player j's defection. This can obstruct player $j$'s deviation, if another player's deviation causes a loss to player $j$. 

\begin{definition}\label{defin:intrinsic probability}
{\rm The \emph{prior} probability that player $j$ deviates from the coalition structure $p$ is
$$
\tau_{i,j}(p):=\frac{D_j(p)}{D_j(p)+R_j(p)}.
$$}
\end{definition}

The terminology \emph{prior} wants to clarify the fact that the event ``\emph{player j abandons the coalition}'' is not measureable in any absolute and meaningful sense. The \emph{prior probability} is a sort of measure a priori of this event knowing only mathematically measurable information, as monetary incentive and monetary risk.

\begin{remark}\label{rem:zero risk}
{\rm If the set $T(\sigma,\sigma_j')$ is empty for all $(\sigma,\sigma_j')\in\text{Dev}_j(p)$, then the supremum defining the risk $R_j(p)$ is equal to zero. Consequently, the prior probability that player $j$ abandons the coalition structure $p$ is equal to 1. This is coherent with the intuition that if $T(\sigma,\sigma_j')=\emptyset$, then there is no way to obstruct player $j$'s defection.}
\end{remark}

As said before, we can now compute all remaining probabilities $\tau_{i,J}(p)$ assuming that the events ``\emph{player $j$ deviates from $p$}'' and ``\emph{player $k$ deviates from $p$}'' are independent. In particular, $\tau_{i,\emptyset}(p)$ will represent the probability that none of the players other than $i$ deviates from the coalition structure.\\

\textbf{Definition of the numbers $e_{i,J}(p)$.}  We recall that the numbers $e_{i,J}(p)$ represent the infimum of gains of player $i$ when the players in $J$ decide to deviate from the coalition structure $p$. Therefore, the definition of these numbers is very straightforward. Let $J\subseteq P\setminus\{i\}$, we first define the set

$$
\text{Dev}_J(p):=\left\{(\sigma,\sigma_J')\in\bigotimes_{\alpha=1}^k M(p_\alpha,p)\times\bigotimes_{j\in J}\mathcal P(S_j) : \exists j\in J : g_j(\pi_j(\sigma_J'),\sigma_{-j})\geq g_j(\sigma)\right\}.
$$

Then we define

$$
e_{i,J}(p):=\inf\{g_i(\sigma'_J,\sigma_{-J}) : (\sigma,\sigma_J')\in\text{Dev}_J(p)\}.
$$

\begin{definition}\label{defin:value}
{\rm The value of the coalition structure $p$ for player $i$ is
\begin{align}
v_i(p)=\sum_{J\subseteq P\setminus\{i\}}e_{i,J}(p)\tau_{i,J}(p).
\end{align} }
\end{definition}
We stress that at this first stage of the research we cannot say that this formula is eventually the right way to compute the value of a coalition structure. It just seems a fairly natural way and, as we will show in Section \ref{se:examples}, it meets experimental data satisfactorily well. However, it is likely that a future research, possibly supported by suitable experiments, will suggest to use of a different formula. For instance, we will describe in Example \ref{ex:generalised coordination game} that it is possible that the deviation $D_j(p)$ should be computed taking into account not only deviation to achieve higher gains, but also to get a safe gain.

Now, in an exact theory, player $i$ is assumed to have unbounded rationality and is assumed not to make mistakes in the computations and so, using the principle of cooperation, she will \emph{play according} to some $p$ which maximises the value function $v_i$. It remains to understand the meaning of \emph{playing according with a coalition structure $p$}. Indeed, we cannot expect that player $i$ will play surely according to an acceptable Nash equilibrium of $\mathcal G_p$, since she knows that other players may deviate from the coalition. What we can do is to use the numbers $v_i(p)$ to define a sort of beliefs.

\begin{definition}
{\rm Let $A\subseteq\mathcal P(S_1)\times\ldots\times\mathcal P(S_N)$. The subgame induced by $A$ is the \emph{game} whose set of mixed strategies of player $i$ is the closed convex hull in $\mathcal P({S_i})$ of the projection set $\pi_i(A)$.}
\end{definition}

Therefore, a subgame induced by a set $A$ is not, strictly speaking, a game, since in general the set of mixed strategies of player $i$ cannot be described as the convex hull of a set of pure strategies which is a subset of $S_i$. In the induced game only particular mixed strategies are allowed, which, as said earlier, correspond to some sort of beliefs. Observe that, since the set of allowed mixed strategies is convex and compact, we can formally find a Nash equilibrium of an induced game. Indeed, Nash's proof of existence of equilibria does not really use the fact that the utility functions are defined on $\mathcal P(S_1)\times\ldots\times\mathcal P(S_N)$, but only that they are defined on a convex and compact subset of $\mathcal P(S_1)\times\ldots\times\mathcal P(S_N)$.

Let $\text{Ind}(\mathcal G,p)$ be the subgame induced by the set of strategies $\sigma\in\mathcal P(S_1)\times\ldots\times\mathcal P(S_N)$ such that $g_i(\sigma)\geq v_i(p)$, for all $i\in P$.  Observe that the induced game is not empty, since $v_i(p)$ is a convex combinations of infima of values attained by the gain function $g_i$.

\begin{definition}{\bf (Exact cooperative equilibrium)}\label{defin:exact cooperative equilibrium eut}
{\rm An exact cooperative equilibrium is one where player $i$ plays a Nash equilibrium of the subgame $\text{Ind}(\mathcal G,p)$ where $p$ maximizes $v_i(p)$ \footnote{Observe that this is well defined also in case of multiple $p$'s maximizing $v_i(p)$, since the induced games $\text{Ind}(\mathcal G,p)$ and $\text{Ind}(\mathcal G,p')$ are the same, if $p,p'$ are both maximizers. }. }
\end{definition}

One could define a quantal cooperative equilibrium, declaring that player $i$ plays with probability $e^{\lambda v_i(p)}/\sum_pe^{\lambda v_i(p)}$ according to the quantal response equilibrium or the quantal level-k theory applied to $\text{Ind}(\mathcal G,p)$. At this first stage of the research, we are not interesting in such refinements, that could be useful in future and deeper analysis (cf. Examples \ref{ex:generalised coordination game} and \ref{ex:asymmetric matching pennies}).

\section{Examples and experimental evidence}\label{se:examples}

In this section we apply the cooperative equilibrium (under expected utility theory and without using altruism) to some well known games. The results we obtain are encouraging, since the predictions of the cooperative equilibrium are always satisfactorily close to the experimental data. We present also two examples where the cooperative equilibrium makes new predictions, completely different from all standard theories. These new predictions are partially supported by experimental data, but we do not have enough precise data to say that they are strongly confirmed.

\begin{example}\label{ex:traveler2}
{\rm Let $\mathcal G^{(b)}$ be the parametrized Traveler's Dilemma in Example \ref{ex:traveler} with bonus-penalty equal to $b$.  Let $p_c=(\{1,2\})$ be the cooperative coalition. We recall that in Example \ref{ex:traveler} we have shown that the profile of strategies $(300,300)$ is the unique acceptable equilibrium for $p_c$. To compute the values of $p_c$, let $i=1$ (the case $i=2$ is the same, by symmetry). One has $D_2(p_c)=b-1$, corresponding to the strategy profile $(300,299)$. Corresponding to this deviation of player 2, which is the unique deviation maximizing player 2's gain, the best deviation for player 1 is to play the strategy $298$, which gives $g_2(300,300)-g_2(298,299)=2+b$. Therefore, $R_2(p_c)=2+b$. Consequently, we have
$$
\tau_{1,\{2\}}(p_c)=\frac{b-1}{2b+1}\qquad\text{and}\qquad\tau_{1,\emptyset}(p_c)=\frac{b+2}{2b+1}.
$$

Now, $e_{1,\{2\}}=300-2b$, corresponding to the profile of strategy $(300,300-b)$, and $e_{1,\emptyset}(p_c)=300$. Consequently, setting $v_1(p_c)=v_2(p_c)=:v(p_c)$, we have
$$
v_(p_c)=300\cdot\frac{b+2}{2b+1}+(300-2b)\cdot\frac{b-1}{2b+1}.
$$
On the other hand, the selfish coalition structure $p_s=(\{1\},\{2\})$ has value 
$$
v_1(p_s)=v_2(p_s)=180,
$$
since there are no possible deviations from a Nash equilibrium. 
Therefore, for small values of $b$, one has $v(p_c)>v(p_s)$ and the cooperative equilibrium predicts that agents play according to the cooperative coalition; for large values of $b$, one has $v(p_s)>v(p_c)$ and then the cooperative equilibrium predicts that agents play the Nash equilibrium. Moreover, the rate of cooperation depends on $b$: the larger is $b$, the smaller is the rate of cooperation predicted. We are aware of only two experimental studies devoted to one-shot Traveler's dilemma. In this cases, the predictions are even quantitatively close.

\begin{itemize}
\item For $b=2$ and $S_1=S_2=\{2,3,\ldots,100\}$, it has been reported in \cite{Be-Ca-Na05} that most of subjects (38 out of 45) chose a number between $90$ and $100$ and the strategy which had the highest payoff was $s=97$. In our case, we obtain
$$
v(p_c)=100\cdot\frac{b+2}{2b+1}+96\cdot\frac{b-1}{2b+1}=99.2.
$$
Consequently, the cooperative equilibrium is supported \emph{near} $99$.
\item For $b=5$ and $S_1=S_2=\{180,181,\ldots,300\}$, it has been reported in \cite{Go-Ho01} that about 80 per cent of the subjects submitted a strategy between 290 and 300, with an average of 295. In our case, we obtain
$$
v(p_c)=300\cdot\frac{b+2}{2b+1}+290\cdot\frac{b-1}{2b+1}=296.35.
$$
Consequently, the cooperative equilibrium is supported \emph{between} $296$ and $297$, which is very close to the experimental data.
\item For $b=180$ and $S_1=S_2=\{180,181,\ldots,300\}$, it was reported in \cite{Go-Ho01} that about 80 per cent of the subjects played the Nash equilibrium 180. In our case, one easily sees that
$$
v(p_c)<v(p_s)
$$
Consequently, the cooperative equilibrium reduced to Nash equilibrium and predicts the solution $(180,180)$. So the cooperative equilibrium coincides with what most subjects played.
\end{itemize}
}

\begin{example}\label{ex:prisoner2}
{\rm We consider the parametrized Prisoner's dilemma as in Example \ref{ex:prisoner}. Observe that all known solution concepts predict either defection for sure or cooperation for sure. Nevertheless, the data collected on the conceptually similar parametrized Traveler's dilemma suggest that human behavior in the parametrized Prisoner's dilemma should depend on the parameter. This intuition is partially supported by the results presented in \cite{DRFN08}, where the authors reported on experiments conducted on the repeated Prisoner's dilemma with punishment and observed that subjects tend to cooperate more when the cost of cooperating is smaller. Motivated by these experimental data, Fudenberg, Rand, and Dreber indeed asked ``\emph{How do the strategies used vary with the gains to cooperation?}'' (cf. \cite{Fu-Ra-Dr12}, p.727, Question 4).

We now show that in fact cooperative equilibrium predicts a rate of cooperation which depends on the particular gains. 

\begin{proposition}\label{prop:prisoner}
The unique cooperative equilibrium of the parametrized Prisoner's dilemma $\mathcal G^{(\mu)}$ is:
\begin{itemize}
\item $(D,D)$ if $\mu\leq1$,
\item $\left(\frac{\mu-1}{\mu}C+\frac{1}{\mu}D,\frac{\mu-1}{\mu}C+\frac{1}{\mu}D\right)$.
\end{itemize}
In particular, the cooperative equilibrium of $\mathcal G^{(\mu)}$ verifies the following appealing property:
\begin{enumerate}
\item It predicts defection for $\mu=0$,
\item It moves continuously and monotonically from defection to cooperation, as $\mu$ increases,
\item It converges to cooperation as $\mu\to\infty$.
\end{enumerate}
\end{proposition}

\begin{proof}
The cooperative coalition structure $p_c=(\{1,2\})$ gives rise to a one-player game whose unique Nash equilibrium is the cooperative profile $(C,C)$. The value of this coalition is, for both players,
$$
v_1(p_c)=v_2(p_c)=(1+\mu)\left(1-\frac{1}{1+\mu}\right)=\mu.
$$ 
The selfish partition $p_s=(\{1\},\{2\})$ gives rise to the classical Nash equilibrium $(D,D)$. The value of $p_s$ is then, for both players,
$$
v_1(p_s)=v_2(p_s)=1
$$
Therefore, for $\mu<1$ one has $v_1(p_c)=v_2(p_c)<v_1(p_s)=v_2(p_s)$ and therefore the cooperative equilibrium predicts defection. To compute the cooperative equilibrium for $\mu\geq1$, first we need to find all profiles of strategies $(\sigma_1,\sigma_2)$ such that

\begin{align}\label{eq:first inequality1}
\begin{cases} g_1(\sigma_1,\sigma_2)\geq\mu\\ g_2(\sigma_1,\sigma_2)\geq\mu\end{cases}
\end{align}
To this end, set $\sigma_1=\lambda_1a_1+(1-\lambda_1)b_1$ and $\sigma_2=\lambda_2a_2+(1-\lambda_2)b_2$. From Equation (\ref{eq:first inequality1}) one gets
\begin{align}\label{eq:second inequality1}
\begin{cases} \lambda_1\lambda_2(1+\mu)+(1-\lambda_1)\lambda_2(2+\mu)+(1-\lambda_1)(1-\lambda_2)\geq\mu\\ \lambda_1\lambda_2(1+\mu)+(1-\lambda_2)\lambda_1(2+\mu)+(1-\lambda_1)(1-\lambda_2)\geq\mu\end{cases}
\end{align}
To compute the Nash equilibrium restricted to the induced game defined by these strategies is very easy. Indeed, it is clear, by simmetry, that this Nash equilibrium must be symmetric and so it is enough to find the lowest $\lambda$ such that $(\lambda,\lambda)$ is a solution of (\ref{eq:second inequality1}). One easily finds $\lambda=\frac{\mu-1}{\mu}$, as claimed.
\end{proof}

As a specific example of a one-shot Prisoner's dilemma, we consider the one recently experimented using MTurk in \cite{DEJR12} with monetary outcomes (expressed in dollars) $T=0.20, R=0.15,P=0.05,S=0$. Fix $i=1$. Denote by $p_s$ the selfish coalition structure, where the players are supposed to act separately. Then $\mathcal G_{p_s}=\mathcal G$, whose unique Nash equilibrium is $(D,D)$. Since a Nash equilibrium has no deviations, then $D_2(p_c)=0$ and consequently $v(p_s)=0.05$. Now, let $p_c$ be the cooperative coalition structure, where the players are supposed to play together. The game $\mathcal G_{p_c}$ is a one-player game whose only Nash equilibrium is $(C,C)$. Now, $D_2(p_c)=0.05$, since the second player can get $0.20$ instead of $0.15$ if she defects and the first player cooperates, and $R_2(p_c)=0.10$, since the second player risks to get $0.05$ instead of $0.15$ if also the other player defects. Finally $e_{1,\emptyset}(p_c)=0.15$ and $e_{1,2}(p_c)=0$. Consequently, $v(p_c)=0.10$, that is larger than $v(p_s)$. So we need to compute the Nash equilibrium of $\text{Ind}(\mathcal G,p_c)$. By symmetry of the game, this is the same as finding the smallest $\lambda$ such that $0.15\lambda^2+0.2\lambda(1-\lambda)+0.05(1-\lambda)^2\geq0.1$, that is $\lambda=\frac12$. Consequenty, the cooperative equilibrium of this variant of the Prisoner's dilemma is $\frac12 C+\frac12 D$ for both players. Notice that in \cite{DEJR12} it has been reported that players cooperated with probability 58 per cent in one treatment and 65 per cent in another treatment and the over-cooperation in the second experiment was explained in terms of framing effect due to the different ways in which the same game were presented.
}

\end{example}

\end{example} 

\begin{example}\label{ex:bertrand}
{\rm Let us consider the Bertrand competition. Each of N players simultaneously chooses an integer between $2$ and $100$. The player who chooses the lowest number gets a dollar amount times the number she bids and the rest of the players get 0. Ties are split among all players who submit the corresponding bid. 

The unique Nash equilibrium of this game is to choose $2$. Nevertheless, it has been reported in \cite{Du-Gn00} that humans tend to choose larger numbers. It was also observed that the claims tend to get closer to the Nash equilibrium, when the number of players gets larger.

To compute the value of the cooperative coalition $p_c=(\{1,\ldots,N\})$ we observe that every player $j$ has incentive $D_j(p_c)=49$ and risk $R_j(p_c)=50$. We then obtain
\begin{itemize}
\item For $N=2$, $v_1(p_c)=v_2(p_c)=50\cdot\frac{50}{99}$,
\item For $N=4$, one has
$$
v_1(p_c)=\ldots=v_N(p_c)=50\cdot\left(1-3\cdot\frac{49}{99}+3\cdot\left(\frac{49}{99}\right)^2-\left(\frac{49}{99}\right)^3\right),
$$
\item and so forth.
\end{itemize}

In other words, using the law of  total probability, one can easily show that the value of the cooperative coalition converges to $0$ very quickly. Consequently, when $N$ increases, the value decreases and the cooperative equilibrium predicts smaller and smaller claims.  This matches qualitatively what reported in a repeated Bertrand competion in \cite{Du-Gn00}. 
}
\end{example}

\begin{example}\label{ex:ultimatum}
{\rm In this example we show that the cooperative equilibrium theory fits an experiment reported by Kahneman, Knetsch and Thaler in \cite{KKT86}. Consider the ultimatum game. A proposer and a responder bargain about the distribution of a surplus of fixed size that we suppose normalized to ten. The responder's share is denoted by $s$ and the proposer's share by $10-s$. The bargaining rules stipulate that the proposer offers a share $s\in[0,10]$ to the responder. The responder can accept or reject $s$. In case of acceptance the proposer receives a monetary payoff $10-s$, while the responder receives $s$. In case of a rejection both players receive a monetary return of zero. 

Kahneman, Knetsch and Thaler conducted the following experiment: 115 subjects, divided in three classes, were asked to say what would be the minimum offer (between $0$ and $10$ Canadian dollars) that they would accept, if they were responders. The mean answers were between 2.00, 2.24 and 2.59 (see \cite{KKT86}, Table 2). 

Now, cooperative equilibrium theory predicts that the responder would accept any offer larger than the value of the coalition structure with the largest value. So let us compute the value for the responder of the two coalition structures $p_s$ and $p_c$ assuming that the two players have the same perception of money.

Denote by $A$ and $R$ responder's actions \emph{accept} and \emph{reject}, respectively. As in Nash bargaining problem, we obtain that the cooperative coalition $p_c=(\{1,2\})$ leads to a one-player game $\mathcal G_{p_c}$ with the unique acceptable equilibrium $(5,A)$. Therefore, we have 

$$v_2(p_c)=\frac52$$ 

since the first player can abandon the coalition playing every $s<\frac52$, but she risks to lose everything if the second player rejects the offer (observe that $R$ is a $2$-deviation to the strategy $s=0$). On the other hand, of course, one has $v(p_s)=0$, corresponding to the equilibrium $(0,R)$.

Consequently, cooperative equilibrium theory predicts that the responder would accept any offer larger than 2.5 dollars, which fits the experimental data reported in \cite{KKT86}.

In a very recent and not yet published experiment, Wells and Rand \cite{We-Ra} reported that the average claim of 44 subjects was 10.7 out of 30 \emph{monetary units}. This corresponds to 35.6 per cent which is apparently quite larger than what cooperative equilibrium predicts. However, making the average between the (normalized) results in \cite{KKT86} and \cite{We-Ra} - 44 subjects claimed an average of 0.356, 43 subjects claimed an average of 0.259, 37 subjects claimed an average of 0.224, and 35 subjects claimed an average of 0.200 - one finds an average claim of 0.264, which is in fact very close to the prediction of the cooperative equilibrium, which is 0.25.  
}
\end{example}

\begin{remark}\label{rem:ultimatum game}
{\rm The cooperative equilibrium can predict well also other experimental data collected for the ultimatum game. 

Recall that the unique subgame perfect equilibrium of the ultimatum game is to offer $s=0$. Nevertheless, there are numerous experimental studies which reject this prediction and show that proposers almost always make substantially larger offers. Fehr and Schmidt \cite{Fe-Sc99} explained these observations making use of two parameters $\alpha_i,\beta_i$ for each player. Let us find out what happens using cooperative equilibrium. 

Concerning the selfish coalition $p_s=(\{1\},\{2\})$. One easily sees that 
$$v_1(p_s)=v_2(p_s)=0,$$ 
in correspondence to the subgame perfect equilibrium $(0,R)$. Concerning the cooperative coalition, we have
$$v_1(p_c)=\frac12,$$ 
since the second player has no incentive to abandon the coalition, and

$$v_2(p_c)=\frac14,$$ 

as shown in Example \ref{ex:ultimatum}. Consequently, the exact cooperative equilibrium predicts that the proposer offers $s=0.25$ and the responder accepts. This explains the fact that there are virtually no offer below $0.2$ and above $0.5$, which was observed in \cite{Fe-Sc99} making a comparison among experimental data collected in \cite{GSS82}, \cite{KKT86}, \cite{FHSS88}, \cite{RPOZ91}, \cite{Ca95}, \cite{HMcS96}, and \cite{Sl-Ro97}. 

So there are some data that can be explained by the cooperative equilibrium under expected utility theory and without altruism. Other data can be explained using altruism. For instance, it was observed that proposer's offer was very often higher than 0.25 and, in most of the cases, it was between $0.4$ and $0.5$ (cf. \cite{Fe-Sc99}, Table I). This stronger deviation towards cooperation is not predicted by the exact cooperative equilibrium without altruism and we will show in Example \ref{ex:ultimatum2} how the cooperative equilibrium with altruism can explain it.}
\end{remark}

We now discuss an example that we believe is relevant because it makes predictions that are significantly different from Nash equilibrium. Such predictions are partially confirmed by experimental data, but it would be important to conduct more precise experiments in order to see how humans behave in such a situation.

\begin{example}\label{ex:public good}
{\rm Let us consider the $N$-player public good game. There are $N$ players, each of which has to decide on her contribution level $x_i\in[0,y]$ to the public good. The monetary payoff of player $i$ is given by
$$
g_i(x_1,x_2,\ldots, x_N)=y-x_i+\alpha(x_1+x_2+\ldots+x_N),
$$

where $\frac1N<\alpha<1$ denotes the constant marginal return to the public good $X=x_1+x_2+\ldots+x_N$. Notice that the unique perfect equilibrium is to choose $x_i=0$. Nevertheless, this \emph{free ride hypothesis} has been rejected by numerous experimental studies (see, e.g., \cite{Ma-Am81}, \cite{Is-Wa88}, \cite{IWW94}, \cite{Le95}). In particular, it was explicitly reported in \cite{Is-Wa88} and \cite{IWW94} the intuitive fact that, for a fixed number of player, claims get larger as $\alpha$ get larger and the much less intuitive fact that, for a fixed $\alpha$, claims get larger when the number of players is large enough. We now show that the first property is predicted by the cooperative equilibrium and we anticipate that the second property is predicted by the cooperative equilibrium under cumulative prospect theory. 

\begin{proposition}\label{prop:public good 1}
Let $N$, the number of players, be fixed. Denote $v(p_c)$ and $v(p_s)$ respectively the value of the cooperative coalition structure $p_c=(\{1,\ldots,N\})$ and of the selfish coalition structure $p_s=(\{1\},\ldots,\{N\})$. Then the function $v(p_c)-v(p_s)$ has the following properties:
\begin{enumerate}
\item it is strictly increasing in the variable $\alpha$,
\item it is negative for $\alpha=\frac1N$,
\item it is positive for $\alpha=1$.
\end{enumerate}
\end{proposition}

The proof of this proposition is a long and tedious computation. Here we report explicitly only the proof for $N=2$. How to treat the general case should then be clear (use the law of total probabilities).

\begin{proof}[Proof of Proposition \ref{prop:public good 1} with $N=2$.]
Let $p_c=(\{1,2\})$ be the cooperative coalition structure. The unique Nash equilibrium of the game $\mathcal G_{p_c}$ is $(y,y)$ and each of the two players gets $e_{1,\emptyset}(p_c)=2\alpha y$. Assume $i=1$ (the case $i=2$ is symmetric). Observe that $D_2(p_c)=y+\alpha y-2\alpha y=y-\alpha y$. Indeed, the best deviation for player $2$ is to play $x_2=0$, which gives a payoff of $y+\alpha y$, if $x_1=y$. The risk is $R_2(p_c)=2\alpha y-y$. Indeed, if also player $1$ abandons the coalition $p_c$ to play the selfish strategy $x_1=0$, player $2$ would get $y$ instead of $2\alpha y$. Consequently
$$
\tau_{1,\{2\}}(p_c)=\frac{y-\alpha y}{y-\alpha y+2\alpha y-y}=\frac{1-\alpha }{\alpha }.
$$
On the other hand, one has
$$
e_{1,\{2\}}(p_c)=\alpha y,
$$
corresponding to player 2's defection.
Therefore,
$$
v_1(p_c)=v_2(p_c)=2\alpha y\cdot\frac{2\alpha -1}{\alpha }+\alpha y\cdot\frac{1-\alpha}{\alpha}=(3\alpha-1)y.
$$
On the other hand, the selfish coalition $p_s=(\{1\},\{2\})$ has value $y$, corresponding to the equilibrium $(0,0)$. Consequently, the function $v(p_c)-v(p_s)$ is strictly increasing in the variable $\alpha$ and one has

$$
v(p_c)=v(p_s) \iff \alpha=\frac23.
$$

\end{proof}

As a quantitative comparison, we consider the experimental data reported in \cite{GHL02}, with $\alpha=0.8$. We normalize $y$ to be equal to 1 (in the experiment $y=0.04$ dollars). In this case the cooperative equilibrium is \emph{supported between} 0.66 and 0.67. In \cite{GHL02} it has been reported that the average of contributions was 0.50, but the mode was 0.60 (6 out of 32 times) followed by 0.80 (5 out of 32 times). 
}
\end{example}

\begin{example}\label{ex:bargaining2}
{\rm We consider the finite version of Nash's bargaining problem as in Example \ref{ex:bargaining}. It is well known that the unique reasonable solution is $(50,50)$ and indeed a number of theories has been developed to select such a Nash equilibrium. For instance, in \cite{Na50b}, \cite{Ka-Sm75}, and \cite{Ka77}, the authors studied a set of additional axioms that guarantee that the unique solution of Nash bargaining problem is a 50-50 share. Other solutions, based on different solution concepts, have been recently proposed in \cite{Ha-Ro10} and \cite{Ha-Pa12}.

Now we show that also the cooperative equilibrium predicts a 50-50 share, if the two players have the same perception of gains.

\begin{proposition}\label{prop:bargaining}
If the two players have the same perception of money, that is, $f_1=f_2$, then the unique exact cooperative equilibrium is $(50,50)$.
\end{proposition}

\begin{proof}
As we have already seen in Example \ref{ex:bargaining}, the cooperative partition $p_c$ has a unique acceptable profile of strategies, which is $\left(50,50\right)$. Observe that $\text{Dev}_j(p)=\emptyset$, for all $j$, and therefore $\text{Ind}(\mathcal G,p)$ is the game where both players can choose only the strategy $50$. Consequently, we have
$$
v_1(p_c)=v_2(p_c)=50.
$$
Now consider the selfish coalition structure $p_s=(\{1\},\{2\}$). This time the unique acceptable equilibria are
$$
\text{Acc}_{\{1\}}(\mathcal G_{p_s})=(100,0)\qquad\text{Acc}_{\{2\}}(\mathcal G_{p_s})=(0,100).
$$
Observing that $\text{Dev}(p_s)=\emptyset$, we then obtain
$$
v_1(p_s)=g_1(100,100)=0.
$$
Analogously, we obtain $v_2(p_s)=g_2(100,100)=0$. Therefore the value of the cooperative coalition structure is larger than the value of the selfish coalition structure  and, consequently, the set of exact cooperative equilibria of Nash bargaining problem coincides with the set of Nash equilibria of the induced game $\text{Ind}(\mathcal G,p_c)$. Since this induced game contains only one profile of strategies, which is $(50,50)$, this is then its unique exact cooperative equilibrium.
\end{proof}
}
\end{example}




We mentioned in the Introduction that there are other solution concepts that have been proposed in the last few years and we have discussed why believe that Renou-Schlag-Halpern-Pass's iterated regret minimization is the most promising of them: the others are either too rigid or inapplicable to one-shot games. Contrariwise, iterated regret minimization can explain deviations from Nash equilibria in several games. Nevertheless, as observed in \cite{Ha-Pa12}, it fails to predict human behavior for some other games, such as the Prisoner's dilemma, the public good game, and the Traveler's dilemma with punishment. We have already computed the cooperative equilibrium for the Prisoner's dilemma and the public good game and we now make a parallelism between iterated regret minimization and cooperative equilibrium for the Traveler's dilemma with punishment.

\begin{example}\label{ex:punish travelers}
{\rm Consider a variant of the Traveler's dilemma that has been proposed in \cite{Ha-Pa12}, Section 6. Let us start from the Traveler's dilemma in Example \ref{ex:traveler} where, this time, the strategy set is $\{2,3,\ldots,100\}$ for both players and the bonus-penalty is $b=2$. Suppose that we modify this variant of the Traveler's dilemma so as to allow a new action, called P (for \emph{punish}), where both players get 2 if they both play P, but if one player plays P and the other plays an action other than P, then the player who plays P gets $2$ and the other player gets $-96$. In this case $(P,P)$ is a Nash equilibrium and it is also the solution in terms of regret minimization. As observed in \cite{Ha-Pa12}, this is a quite unreasonable solution, since the intuition suggests that playing $P$ should not be rational. In fact, one can easily check that, from our point of view, this game is absolutely the same as the original Traveler's dilemma\footnote{Basically because strategies with very small payoff, such as $P$, do not enter in our computation of the value of the cooperative coalition.} and therefore it has got the same cooperative equilibria.}
\end{example}









\section{Towards cumulative prospect theory}\label{se:towards cpt}

In the previous section we have discussed a set of examples where the cooperative equilibrium under expected utility theory predicts human behavior satisfactorily well. On the other hand, since we are working with gain functions, it is natural to use cumulative prospect theory instead of expected utility theory. But before describing the cooperative equilibrium under cumulative prospect theory, we discuss a few examples where the passage from expected utility theory to cumulative prospect theory may explain observations that are not consistent with the cooperative equilibrium under expected utility theory. 

\begin{example}\label{ex:public goods}
{\rm We mentioned before that has been observed that contributions in the Public Goods game depend on the number of players in a puzzling way: they first decreases as the number of players increases, but then, when the number of players if sufficiently large, they increase again. This behavior is not predicted by the cooperative equilibrium under expected utility theory, which predicts that contributions decreases as the number of players increases. Nevertheless, this behavior is consistent with the cooperative equilibrium under cumulative prospect theory.

Indeed, given the $N$-player Public Goods game with marginal return $\alpha$, the prior probability that player $j$ abandons the coalition is

$$
\tau_{i,j}(p_c)=\frac{1-\alpha}{\alpha(N-1)}.
$$

Consequently, when $N$ is large enough, all the events ``\emph{$j$ abandons the coalition}'' have negligible probability. Now, one of the principles of cumulative prospect theory is that decision makers treat extremely unlikely events as impossible (see \cite{Ka-Tv79}, p.275) and therefore, a part from very risk averse people, most of the agents would actually replace this probability just by $0$. So the cooperative equilibrium is consistent with the tendency to cooperate that has been observed in large groups. 

}
\end{example}

\begin{example}\label{ex:generalized coordination game 2}
{\rm The following game has been proposed by J. Halpern in a private communication. Two players have the same strategy set $\{a,b,c\}$ and the gains are described by the matrix

$$
  \begin{array}{cccc}
     & a & b & c\\
    a & x,x & 0,0 & 0,y\\
    b & 0,0 & x,x & 0,y\\
    c & y,0 & y,0 & y,y \\
  \end{array}
$$
where $x>y>0$. In this case one finds $v(p_c)=v(p_s)=0$ and consequently, the set of exact cooperative equilibrium is equal to the set of Nash equilibria. Nevertheless, in this case it is very likely that if $y$ and $x$ are very close and much larger than $0$, then the two players should coordinate and play the safe strategy $c$. Also this behavior would be predicted by the cooperative equilibrium under cumulative prospect theory: the strategies $a$ and $b$ are deleted a priori since perceived too risky with respect to the safe strategy. 
}
\end{example}

\section{A brief introduction to cumulative prospect theory}\label{se:cpt}

The examples described in the previous sections give one more motivation to abandon expected utility theory and use cumulative prospect theory. Before starting the description of the cooperative equilibrium under cumulative prospect theory, we take this short section to give a short introduction to this theory.\\

By definition, a prospect $p=(x_{-m},p_{-m};\ldots;x_{-1},p_{-1};x_0,p_0;x_1,p_1;\ldots;x_n,p_n)$ yields outcomes\footnote{Prospect theory and cumulative prospect theory have been originally developed for monetary outcomes (see \cite{Ka-Tv79}, p.274, l.4), giving us one more motivation to abandon utility functions and work with gain functions. Kahneman and Tversky's choice to work with monetary outcomes is probably due to the second principle of their theory, as it will be recalled little later.} $x_{-m}<\ldots<x_{-1}<x_0=0<x_1<\ldots<x_n$ with probabilities $p_i>0$, for $i\neq0$, and $p_0\geq0$, that sum up to $1$.

Expected utility theory was founded by Morgenstern and von Neumann in \cite{Mo-vN47} to predict the behavior of a decision maker that must choose a prospect among some. Under certain axioms (see, for instance, \cite{Fi82}) Morgenstern and von Neumann proved that a decision maker would evaluate each prospect $p$ using the value
\begin{align}\label{eq:eut}
V(p)=\sum_{i=-m}^np_iu(x_i),
\end{align}
where $u(x_i)$ is the utility of the outcome $x_i$, and then she would choose the prospect(s) maximizing $V(p)$.

It has been first realized by M. Allais in \cite{Al53} that a human decision maker does not really follow the axioms of expected utility theory and, in particular, she evaluates a prospect using an evaluation procedure different from the one in (\ref{eq:eut}). A first attempt to replace expected utility theory with a theory founded on different axioms and able to explain deviations from \emph{rationality} was done in \cite{Ka-Tv79}, where Kahneman and Tversky founded the so-called prospect theory. This novel theory encountered two problems. First, it did not always satisfy stochastic
dominance, an assumption that many theorists were reluctant to give up. Second, it was not
readily extendable to prospects with a large number of outcomes. Both
problems could be solved by the rank-dependent or cumulative functional, first proposed
by Quiggin \cite{Qu82} for decision under risk and by Schmeidler\cite{Sc89} for decision under
uncertainty. Finally, Kahneman and Tversky were able to incorporate the ideas presented in \cite{Qu82} and \cite{Sc89} and developed their cumulative prospect theory in \cite{Tv-Ka92}. Prospect theory and cumulative prospect theory have been successfully applied to explain a large number of phenomena that expected utility theory was not able to explain, as the disposition effect \cite{Sh-St85}, asymmetric price elasticity \cite{Pu92},\cite{Ha-Jo-Fa93}, tax evasion \cite{Dh-No07}, as well as many problems in international relations \cite{Le92}, finance \cite{Th05}, political science \cite{Le03}, among many others\footnote{The two papers in prospect theory and cumulative prospect theory have more than 30000 citations.}.

The basic principles of cumulative prospect theory are the following.
\begin{enumerate}
\item[(P1)] Decision makers weight probabilities in a non linear manner. In particular, the evidence suggests that decision makers overweight low probabilities and underweight high probabilities.
\item[(P2)] Decision makers think in terms of gains and losses rather than in terms of their net assets\footnote{This principle is probably the one which forced Kahneman and Tversky to work with monetary outcomes and force us to work with gain functions.}.
\item[(P3)] Decision makers tend to be risk-averse with respect to gains and risk-acceptance with respect to losses\footnote{As a consequence, risk aversion is already taken into account and this is why we did not need to consider it explicitly in the definition of a game in explicit form.}.
\item[(P4)] \emph{Losses loom larger than gains}; namely, the aggravation that one experiences in losing a sum of money appears greater than the pleasure associated with gaining the same amount of money.
\end{enumerate}

The consequence of these principles is that decision makers evaluate a prospect $p$ using a value function
\begin{align}\label{eq:cpt}
V(p)=\sum_{j=-m}^n\pi_jv(x_j)
\end{align}
that is completely different from the one in (\ref{eq:eut}). To understand the explicit shape of the functions $v$ and $\pi$ is probably the most important problem in cumulative prospect theory. About the function $v$, it has been originally proposed in \cite{Tv-Ka92} to use the function
$$
v(x)=\left\{
  \begin{array}{ll}
    x^{\alpha}, & \hbox{if $x\geq0$;} \\
    -\lambda(-x)^{\beta}, & \hbox{if $x<0$.}
  \end{array}
\right.
$$
where experiments done in \cite{Tv-Ka92} gave the estimations $\alpha\sim\beta\sim0.88$ and $\lambda\sim2.25$. About the function $\pi$, the situation is much more intrigued: cumulative prospect theory postulates the existence of a strictly increasing surjective function $w:[0,1]\to[0,1]$ such that

$$
\pi_{-m}=w(p_{-m})
$$
$$
\pi_{-m+1}=w(p_{-m}+p_{-m+1})-w(p_{-m})
$$
$$
\vdots
$$
$$
\pi_{j}=w\left(\sum_{i=-m}^jp_i\right)-w\left(\sum_{i=-m}^{j-1}p_i\right)\qquad\qquad j<0
$$
$$
\pi_0=0
$$
$$
\pi_j=w\left(\sum_{i=j}^np_i\right)-w\left(\sum_{i=j+1}^{n}p_i\right)\qquad\qquad j>0
$$
$$
\vdots
$$
$$
\pi_{n-1}=w(p_{n-1}+p_n)-w(p_n)
$$
$$
\pi_n=w(p_n)
$$

A first proposal of such a function $w$ was made by Tversky and Kahneman themselves in \cite{Tv-Ka92} and it is
$$
w(p)=\frac{p^{\gamma}}{\left(p^{\gamma}+(1-p)^{\gamma}\right)^{\frac{1}{\gamma}}}
$$
where $\gamma$ has been estimated to belong to the interval $\left[\frac{1}{2},1\right)$ in \cite{Ri-Wa06}. Other functions $w$ have been proposed in \cite{Ka79}, \cite{Go-Ei87}, 
\cite{R87}, \cite{Cu-Sa89}, \cite{La-Ba-Wi92}, \cite{Lu-Me-Ch93}, \cite{He-Or94}, \cite{Pr98}, and \cite{Sa-Se98}. 

It is not our purpose to give too many details about the enormous literature devoted to understanding the evaluation procedure in cumulative prospect theory. Our purposes were indeed to give a brief introduction to the theory and stress how this theory implies the necessity to work with gain functions instead of utility functions. So we now pass to the description of the cooperative equilibrium for finite games in explicit form under cumulative prospect theory and taking into account altruism.

\section{Iterated Deletion: the set of playable strategies}\label{se:playable}

The cooperative equilibrium under cumulative prospect theory and taking into account altruism will be defined through two steps. In the first step we use the altruism functions $a_{ij}$ to eliminate the strategies that are not good for the collectivity. The second step is the prospect theoretical analogue of the procedure described in Section \ref{se:cooperative equilibrium eut}, applied to the subgame obtained after eliminating the strategies in the first step. 

In this section we describe the first step of the construction, that we call \emph{iterated deletion}. As well known, iterated deletion of strategies is a procedure which is common to most solution concepts (in Nash theory, one deletes dominated strategies; in iterated regret minimization theory, one deletes strategies which do not minimize regret; in Bernheim's and Pearce's rationability theory (\cite{Be84} and \cite{Pe84}), one deletes strategies that are not \emph{justifiable}  \cite{Os-Ru94}). However, the use of altruism to delete strategies seems new in the literature. This iterated deletion of strategies is based on a new notion of domination between strategies, that we call super-domination\footnote{A slightly stronger notion of domination between strategies has been independenlty introduced in \cite{Ha-Pa13}, under the name \emph{minimax domination}.}, which is motivated by the fact that human players do not eliminate weakly or strongly dominated strategies (as shown by the failure of the classical theory to predict human behavior in the Prisoner's and Traveler's Dilemmas).

Each step of our iterated deletion of strategies is made by two sub-reductions. The first sub-reduction is based on the following \emph{principle}:
\begin{itemize}
\item[(CS)] If $s_i\in S_{i}$ is a strategy for which there is another strategy $s_i'\in S_{i}$ which gives a certain larger gain (or a certain smaller loss) to player $i$ and does not harm too much the other players, then player $i$ will prefer the strategy $s_i'$ and will never ever play the strategy $s_i$. 
\end{itemize}
Thus, this principle states that every player is selfish unless the society gets a big damage. As we mentioned before, implicit in this principle there is a new notion of domination between strategies.

\begin{definition}\label{defin:sure gain}
{\rm Let $s_i,s_i'\in S_{i}$. We say that $s_i$ is \emph{super-dominated} by $s_i'$ and we write $s_i<_{i}s_i'$, if
\begin{enumerate}
\item 
 for all $s_{-i},s_{-i}'\in S_{-i}$, one has $g_{i}\left(s_i,s_{-i}\right)\leq g_i\left(s_i',s_{-i}'\right)$,
\item there are $s_{-i},s_{-i}'\in S_{-i}$ such that $g_{i}\left(s_i,s_{-i}\right)< g_i\left(s_i',s_{-i}'\right)$.
\end{enumerate} }
\end{definition}

Observe that super-domination is much stronger than the classical notion of weak domination. This makes sense since it has been observed that in many situations, as in the Traveler's dilemma, players do not eliminate weakly dominated strategies, while it is clear that a purely selfish player would delete a super-dominated strategy. On the other hand, there is no direct relation between super-domination and strong-domination, as shown by the following examples.

\begin{example}\label{ex:sure}
{\rm
Consider the following version of the Prisoner's dilemma
$$
  \begin{array}{ccc}
     & L & R \\
    U & 2,2 & 0,3 \\
    D & 3,0 & 1,1 \\
  \end{array}
$$
The strategy D strongly dominates U and the strategy R strongly dominates the strategy L. Nevertheless, there are no super-dominated strategies, since $g_1(D,R)<g_1(U,L)$ and $g_2(D,R)<g_2(U,L)$.}
\end{example}

\begin{example}\label{ex:sure gain}
{\rm Consider the two-person zero-sum game
$$
  \begin{array}{ccc}
     & L & R \\
    U & 0,0 & 10,-10 \\
    D & 1,-1 & 1,-1 \\
  \end{array}
$$
In this case L super-dominates R, but R is not strongly dominated by L, since $g_2(D,R)=g_2(D,L)$.}
\end{example}

We will see in Example \ref{ex:more iterations} that the notion of super-domination between strategies can be interesting in itself, since it allows to explain some phenomena that are not easy to capture making use of weakly and strongly dominated strategies.

Before coming back to the theory, we need to fix some terminology. Fix $\sigma_i\in\mathcal P(S_i)$, the \emph{fiber game} defined by $\sigma_i$ is the $(N-1)$-player game $\mathcal G_{\sigma_i}$ obtained by $\mathcal G$ assuming that player $i$ plays the strategy $\sigma_i$ surely. Formally, $\mathcal G_{\sigma_i}=\mathcal G(P\setminus\{i\},S_{-i},\mathfrak g,g_{\sigma_{i}},a_{-i},f_{-i})$, where $g_{\sigma_i}$ is the $(N-1)$-dimensional vector whose components are the functions $g_j(\sigma_i,\cdot)$, with $j\in P\setminus\{i\}$, $a_{-i}=(a_{jk})_{j,k\in P\setminus\{i\},j\neq k}$, $f_{-i}=(f_{j})_{j\in P\setminus\{i\}}$. Using a trick which is conceptually similar to the one used in \cite{Ha-Ro10}, we define the cooperative equilibrium by induction on the number of players. 

\begin{definition}\label{defin:one player}
{\rm The cooperative equilibria of a one-player game are all probability measures supported on the set of pure strategies that maximize the gain function and give rise to acceptable equilibria\footnote{As observed in Remark \ref{rem:barycenter}, when there are many such equilibria, it might make sense to consider only the barycenter.}.}
\end{definition}

Now we suppose that we have already defined the cooperative equilibrium for all $(N-1)$-player games and we define the cooperative equilibrium for all $N$-player games. We denote by $\text{Coop}(\mathcal G)$ the set of cooperative equilibria of a game $\mathcal G$.

Now, fix $i\in P$ and let $s,t\in S_{i}$, with $s<_{i}t$. If player $i$ is believed to play the strategy $t$, the other players would answer playing an \emph{equilibrium} of the fiber game $\mathcal G_t$. Since the fiber game has $N-1$ players, we may use the inductive hypothesis. We define the \emph{set of losers} $L_{i}(s,t)$ to be the set of players $j\in P$ such that
\begin{enumerate}
\item $g_j\left(s,\sigma_{-i}^{(s)}\right)>g_j\left(t,\sigma_{-i}^{(t)}\right), \text{ for all }\sigma_{-i}^{(t)}\in\text{Coop}(\mathcal G_t), \sigma_{-i}^{(s)}\in\text{Coop}(\mathcal G_s)$, and
\item $g_j\left(t,\sigma_{-i}^{(t)}\right)<g_i\left(t,\sigma_{-i}^{(t)}\right), \text{ for all }\sigma_{-i}^{(t)}\in\text{Coop}(\mathcal G_t)$. 
\end{enumerate}
In words, $L_{i}(s,t)$ is the set of players that have a certain disadvantage when player $i$ decides to play the strategy $t$ instead of her worse strategy $s$ (Condition (1)) and that are weaker than player $i$ when she plays her better strategy $s$ (Condition (2)).

Now, if player $i$ decides to renounce to play $t$ and accept to play $s$, then she renounce to a certain gain of $\inf\left\{g_i\left(t,\sigma_{-i}^{(t)}\right) : \sigma_{-i}^{(t)}\in\text{Coop}(\mathcal G_t)\right\}$, to accept a smaller gain. Her maximal loss is then:

$$
P_i(s,t):=\sup\left\{\inf\left\{g_i\left(t,\sigma_{-i}^{(t)}\right) : \sigma_{-i}^{(t)}\in\text{Coop}(\mathcal G_t)\right\}-g_i\left(s,\sigma_{-i}^{(s)}\right) : \sigma_{-i}^{(s)}\in\text{Coop}(\mathcal G_s)\right\}.
$$

On the other hand, the best that can happen to player $j\in L_i(s,t)$ if player $i$ decides to play her worse strategy $s$ is

$$
Q_j(s,t):=\sup\left\{g_j\left(s,\sigma_{-i}^{(s)}\right)-\inf\left\{g_j\left(t,\sigma_{-i}^{(t)}\right) : \sigma_{-i}^t\in\text{Coop}(\mathcal G_t)\right\} : \sigma_{-i}(s)\in\text{Coop}{\mathcal G_s}\right\}.
$$

Now, set

$$
P_i'(t):=\inf\left\{g_i\left(t,\sigma_{-i}^{(t)}\right) : \sigma_{-i}^{(t)}\in\text{Coop}(\mathcal G_t)\right\}.
$$
In words, this number is the certain gain that player $i$ would get if she decides to play her better strategy $t$. Now, set
$$
Q_j'(s):=\inf\left\{g_j\left(s,\sigma_{-i}^{(s)}\right) : \sigma_{-i}^{(s)}\in\text{Coop}(\mathcal G_s) \right\}.
$$
In words, this number is the certain gain that player $j$ would get if player $i$ decides to play her worse $s$. 

Therefore, we have reduced the problem of choosing $s$ or $t$ to the following problem: does player $i$ accept to renounce to $P_i(s,t)$ out of $P_i'(s,t)$ in order to give a gain of $Q_j(s,t)$ to player $j$, who already had a certain gain of $Q_j'(s,t)$? This is in fact a generalized dictator game. So, we set 
$$
A_{ij}(s,t)=a_{ij}\left(\frac{Q_j(s,t)}{P_i(s,t)},P_i'(t), Q_j'(s)\right),
$$
and we give the following definition.

\begin{definition}\label{defin:unplayable first type}
{\rm A strategy $s\in S_i$ is \emph{unplayable of the first type} for player $i$ if there is another strategy $t\in S_i$ such that 
\begin{itemize}
\item $s<_{i}t$
\item for all $j\in L_{i}(s,t)$, one has $P_i(s,t)>A_{ij}(s,t)$.
\end{itemize}
In this case we write $s<_i^It$.
 }
\end{definition}

\begin{example}
{\rm Consider the game with gain matrix
$$
  \begin{array}{ccc}
     & L & R \\
    U & 1,1 & 1,1 \\
    D & 1,1 & 2,1 \\
  \end{array}
$$
Observe that $U<_{1}D$. Moreover, $L_{1}(U,D)=\emptyset$ and therefore the second condition in Definition \ref{defin:unplayable first type} is true for trivial reasons. Consequently, the strategy $U$ is unplayable of the first type for the first player. This happens, roughly speaking, because the column-player, playing $D$, can have a gain without damaging the row-player. }
\end{example}

\begin{example}
{\rm A little less trivial example is given by the game represented by the following gain matrix
$$
  \begin{array}{ccc}
     & L & R \\
    U & 0,0 & 0,0 \\
    D & 1,-1 & 1,-1 \\
  \end{array}
$$
Assume $a_{12}(1,1,-1)<1$. Of course, $U<_{1}D$. Now, observe that $P_{1,2}(U,D)=1$ and that $A_{12}(s,t)=a_{12}(1,1,-1)$, thus the strategy U is unplayable of the first type for the vertical player. Roughly speaking, this happens because the vertical player, playing D, will get a certain gain giving a damage to the horizontal player that is small compared to her gain.}
\end{example}

Coming back to the theory, we would like to delete unplayable strategies of the first type. To this end, we need to prove a simple lemma. Given $s_i\in S_i$, let $\text{Maj}^I(s_i)=\left\{s_i'\in S_i : s_i <_i^I s_i'\right\}$.

\begin{lemma}\label{lem:first type maximal}
For all $i\in P$, there exists $s_i\in S_i$ such that $\text{Maj}^I(s_i)=\emptyset$.
\end{lemma}

\begin{proof}
By contradiction, let $\text{Maj}^I(s_i)\neq\emptyset$, for all $s_i\in S_i$. Fix $s_i^{(1)}\in S_i$. An iteration of the property $\text{Maj}^I\neq\emptyset$ allows to construct a chain
$$
s_i^{(1)}<_i^Is_{i}^{(2)}<_i^I\ldots<_i^Is_{i}^{(n)}
$$
By finiteness of the set $S_i$, we may assume that at some point we get $s_{i}^{(n)}=s_{i}^{(1)}$, with $s_{i}^{(n-1)}\neq s_{i}^{(1)}$. Observe that the relation $<_i^I$ might not be transitive, but the underlying relation $<_i$ is transitive. Therefore, we have gotten
$$
s_i^{(1)}<_is_{i}^{(n-1)}\qquad\text{and}\qquad s_i^{(n-1)}<_is_i^{(1)}
$$
that contradict each other.
\end{proof}

Let UnPl$_i^{(1)}(\mathcal G)$ be the set of player $i$'s unplayable strategies of the first type and denote by Pl$_{i}^{(1)}(\mathcal G):=S_{i}\setminus$UnPl$_{i}^{(1)}(\mathcal G)$, that is well defined and non-empty by Lemma \ref{lem:first type maximal}. The notation Pl$_{-i}^{(1)}(\mathcal G)$ stands for the cartesian product of all the Pl$_{j}^{(1)}(\mathcal G)$'s but Pl$_i^{(1)}(\mathcal G)$. \\

Now we start the description of the second sub-restriction, that will be done through the definition of unplayable strategies of the second type. The \emph{principle} underlying this second restriction is somehow the dual principle of the one underlying the previous restriction:
\begin{itemize}
\item[(PA)] If $s\in S_{i}$ is a strategy for which there is another strategy $t\in S_{i}$ such that player $i$ has a little disadvantage, but the other players have a big advantage, then player $i$ will prefer the strategy $t$ in order to help the society.
\end{itemize}

As said earlier, the principle (CS) is a sort of controlled selfishness, whereas the principle (PA) sounds more like \emph{pure altruism}. We can formalize it in a similar way as we formalized (CS).  Indeed, we can use the number $P_{i}(s,t)$ and $A_{ij}(s,t)$ in the dual way.

\begin{definition}\label{defin:unplayable second type}
{\rm A strategy $t\in $Pl$_{i}^{(1)}(\mathcal G)$ is called \emph{unplayable of the second type} for player $i$ if there is another strategy $s\in$ Pl$_{i}^{(1)}(\mathcal G)$ such that 
\begin{enumerate}
\item $s<_{i}t$,
\item There exists $j\in L_{i}(s,t)$ such that $P_i(s,t)\leq A_{ij}(s,t)$.
\end{enumerate} }
\end{definition}

\begin{example}\label{ex:dictator}
{\rm Consider the standard dictator game $\text{Dict}(1,10,0)$, that is, a proposer offers a division of 10 dollars, which the responder has to accept. The standard perfect equilibrium analysis of this games is that the proposer should keep all the money, since the responder has no say. Nevertheless, in experiments has been reported that most \emph{proposers} offer a certain amount of money to the responder (see, for instance,  \cite{Fo-Ho-Sa-Se94}). Bolton and Ockenfels explained this anomalous behavior using equity in \cite{Bo-Oc00}. We can explain it using iterated deletion of strategies using altruism. Let us model the set of strategies of the proposer, for simplicity, by $S=\{0,1,\ldots,10\}$. It is clear that there is a chain of super-dominated strategies for the proposer: $0<_{\text{prop}}1<_{\text{prop}}2<_{\text{prop}}\ldots<_{\text{prop}}10$. Now, one can easily show that every strategy $s$ with $s<a_{\text{prop},\text{resp}}(1,10,0)$ is unplayable of the second type for the proposer. Therefore, cooperative equilibrium theory predicts that the proposer offers a fairer division because of altruism. Moreover, the larger is $a_{\text{prop},\text{resp}}(1,10,0)$, the larger is the offer. 
}
\end{example}

\begin{example}\label{ex:ultimatum2}
{\rm We have seen in Example \ref{ex:ultimatum} that the cooperative equilibrium without altruism of the Ultimatum game is that the proposer offers 0.25 and the responder accepts. Nevertheless, it has been reported that most of proposers actually propose a share closer to $0.5$. This can be explained taking into account altruism. Indeed, if we model the set of strategies of the proposer using the set $S=\{0.00,0.01,0.02,\ldots,1.00\}$, then in the induced game $\text{Ind}(G,p_c)$, the strategy $0.25$ is super-dominated for the proposer by $0.26$, which is super-dominated by $0.27$ and so forth. As in the previous example, some of these strategies are unplayable of the second type and therefore, altruism can explain why offers are tipically larger than 0.25.}
\end{example}

Let UnPl$_{i}^{(2)}(\mathcal G)$ be the set of player $i$'s unplayable strategies of the second type and denote by  Pl$_i^{(2)}(\mathcal G):=$Pl$_{i}^{(1)}(\mathcal G)\setminus$UnPl$_{i}^{(2)}(\mathcal G)$. This set is well defined and non-empty thanks to the obvious analogue of Lemma \ref{lem:first type maximal}. 

Now, we start an iteration of this procedure: we consider the subgame $\mathcal G_2$ of $\mathcal G$ defined by the strategy sets Pl$_{i}^{(2)}(\mathcal G)$ and we reduce again these strategy sets computing the unplayable strategies of the two types; in this way, we get other sets of \emph{playable} strategies Pl$_{i}^{(2)}(\mathcal G_2)$; and we start again the procedure. By finiteness of the strategy sets $S_i$, this iteration stabilizes, that is, at some step $k$, one has have Pl$_{i}^{(2)}(\mathcal G_{k})=$Pl$_{i}^{(2)}(\mathcal G_{k+1})$ and this set is clearly non-empty. We set Pl$_{i}:=$ Pl$_i^{(2)}(\mathcal G_k)$.

\begin{definition}\label{defin:playable strategies}
{\rm The set Pl$_{i}$ is called \emph{set of playable strategies} of player $i$.}
\end{definition}

Before starting the second step of the construction, that is, the prospect theoretical analogue of Section \ref{se:cooperative equilibrium eut}, we give more details about the game introduced in Example \ref{ex:sure gain}. Indeed, this game seems interesting from several viewpoints. First, it is an example where the procedure of elimination of unplayable strategies stabilizes after more than one step. Then, it is one more example where iterated regret minimization theory fails to predict the intuitively right behavior, whereas the cooperative equilibrium does apparently the right job. Finally, it is an example where super-dominated strategies turn out to be helpful to modify iterated regret minimization theory allowing prior beliefs and consequently obtaining the right prediction also under iterated regret minimization theory. 

\begin{example}\label{ex:more iterations}
{\rm Consider the same two-person zero-sum game as in Example \ref{ex:sure gain}, that is, the game with gain matrix
$$
  \begin{array}{ccc}
     & L & R \\
    U & 0,0 & 10,-10 \\
    D & 1,-1 & 1,-1 \\
  \end{array}
$$

Assume that $a_{12}(1,1,-1)<1$. Observe that L super-dominates R and that $L_2(R,L)=\emptyset$. Consequently, R is unplayable of the first type. On the other hand, in this first step U and D are not ordered and therefore, the first step of the iterated deletion leads to the subgame $\mathcal G_2$ where the vertical player still has both strategies U and D available, whereas the horizontal player has only the strategy L. Therefore, in the game $\mathcal G_2$, the strategy $U$ is unplayable of the second type for the column-player (since $a_{12}(1,1,-1)<1)$ and, consequently, one more application of deletion of unplayable strategies leads to the \emph{trivial game} where the vertical player has only the strategy D and the horizontal player has only the strategy L. Therefore, (D,L) is the unique cooperative equilibrium of this game. Observe that this is also a Nash equilibrium. The other Nash equilibrium is $\left(D,\frac{9}{10}L+\frac{1}{10}R\right)$, as one can easily check, which is quite unreasonable, since there is no reason why the horizontal player should play R: playing L she will certainly get at least the same as playing R. Therefore, the cooperative equilibrium coincides with the most reasonable Nash equilibrium. 

On the other hand, a direct application of the iterated regret minimization procedure predicts that the vertical player plays U surely. This is also quite unreasonable, because playing U makes sense only if the column-player plays R. This cannot happen, above all if the column-player understands that the row-player is going to play U. As suggested by Halpern in a private communication, one can fix this problem allowing prior beliefs, in a conceptually similar way as in \cite{Ha-Pa12}, Section 3.5: first one eliminates weakly dominated strategies, then applies iterated regret minimization. Nevertheless, this procedure is questionable on one point: it is not clear why one should eliminate weakly dominated strategies in this context and not in the Traveler's dilemma\footnote{If one eliminates weakly dominates strategies in the Traveler's dilemma before applying iterated regret minimization, one obtains the Nash equilibrium.}. One can fix this problem using super-domination. If one eliminates super-dominated strategies in the game under consideration before applying iterated regret minimization, one finds the \emph{right} solution (D,L), coherently with the classical theory and the cooperative equilibrium. Moreover this is perfectly coherent with the other examples discussed in \cite{Ha-Pa12} and in particular with the Traveler's dilemma: the Traveler's dilemma has many weakly dominated strategies, but none of them is super-dominated.

}
\end{example}


\section{The cooperative equilibrium under cumulative prospect theory}\label{se:cooperative equilibrium cpt}

In this section we finally define the cooperative equilibrium for games in explicit form $\mathcal G=\mathcal G(P,S,\mathfrak g,g,a,f)$ in complete generality. 

In the previous section we have restricted the sets of pure strategies and we have defined the sets of playable strategies $\text{Pl}_i$. We denote by $\text{Red}(\mathcal G)$ this \emph{reduced game}, that is, the subgame of $\mathcal G$ defined by the strategy subsets $\text{Pl}_i$. The cooperative equilibrium of $\mathcal G$ (under prospect theory and taking into account altruism) will be obtained by applying the construction described in Section \ref{se:cooperative equilibrium eut} to the reduced game $\text{Red}(\mathcal G)$ and making use of cumulative prospect theory. To this end, notice that the construction presented in Section \ref{se:cooperative equilibrium eut} depends on expected utility theory only on two points:

\begin{enumerate}
\item We have used expected utility theory to compute the value of the prospect

$$
(e_{i,J}(p),\tau_{i,J}(p))
$$

indexed by $J\subseteq P\setminus\{i\}$. Using cumulative prospect theory, the value that we denoted $v_i(p)$ should be replaced by its prospect theoretical analogue

\begin{align}
v_i^{\text{CPT}}(p)=\sum_{J\subseteq P\setminus\{i\}}v(e_{i,J}(p))\pi_{\tau_{i,J}(p)}.
\end{align}

Since the value $v(x)$ represents how the players perceive a gain of $x$, also the definition of the induced game should be modified: indeed we should allow only the profiles of strategies $\sigma$ such that $v(g_i(\sigma))\geq v_i^{\text{CPT}}(p)$.
Consequently, the two applications of the function $v$, the first in the computation of $v_i^{\text{CPT}}$ and the second in the definition of the induced game, are somehow inverse. Indeed, if $v$ were linear and increasing, the induced game would have been the same as the one obtained by setting $v(x)=x$. Now, we know from cumulative prospect theory that $v$ is strictly increasing. Approximating it by a linear function we can simplify a lot the definition setting $v(x)=x$. This explains why the examples in Section \ref{se:towards cpt} fit the experimental data very well: they have been conducted with relatively small monetary outcomes and there were no possible losses. Of course, it is predictable that in case of possible large gains and/or losses, this approximation will create problems.

\item The definition of the value of a coalition and then the definition of the cooperative equilibrium rely in the computation of Nash equilibria of the games $\mathcal G_p$ and $\text{Ind}(\mathcal G,p)$. The computation of Nash equilibria uses expected utility theory, precisely in the definition of the mixed extension of the gain functions. Unfortunately, the natural translation of Nash equilibrium in the language of cumulative prospect theory leads to define an object that might not exist (see \cite{Cr90} and, more generally, \cite{Fi-Pa10}). To avoid this problem we consider a solution concept which is a bit more general than Nash equilibrium, the so-called equilibrium in beliefs, introduced by Crawford in \cite{Cr90}. Crawford's equilibria in beliefs have the good property to exist in our context, contain all Nash equilibria, and reduce to Nash equilibria in many cases. The remainder of the section is devoted to this.
\end{enumerate}

Before recalling the definition of an equilibrium in beliefs, we need to do a preliminary step, that is writing the mixed extension of the gain functions in the language of cumulative prospect theory. Since notation will get complicated very soon, we start by an example.

\begin{example}\label{ex:from eut to pct}
{\rm Consider the (already reduced) game with gain matrix:
$$
  \begin{array}{ccc}
     & C & D \\
    C & 2,2 & 0,3 \\
    D & 3,0 & 1,1 \\
  \end{array}
$$
Assume that the column-player (player 1) plays the mixed strategy $\sigma_1=\frac{1}{8}C+\frac{7}{8}D$ and player $2$ plays the mixed strategy $\sigma_2=\frac{1}{4}C+\frac{3}{4}D$. Under expected utility theory, we would have
$$
g_1(\sigma_1,\sigma_2)=\sum_{x\in S_1}\sum_{y\in S_2}g_1(x,y)\sigma_1(x)\sigma_2(y).
$$
Let us compute step by step this number to put in evidence where and how expected utility theory must be replaced by cumulative prospect theory. Fix $\sigma_2$ as before and observe that we have a finite family of prospects, one for each pure strategy of the first player. In this example, they are:
$$
p^{(C,\sigma_2)}=\left(2,\frac{1}{4}; 0,\frac{3}{4}\right)\qquad\qquad\text{and}\qquad\qquad p^{(D,\sigma_2)}=\left(3,\frac{1}{4}; 1,\frac{3}{4}\right).
$$
Now, under expected utility theory (and this is the first point where expected utility theory is used), one computes the values of the two prospects, obtaining, in this particular example, the values
$$
V_1(C,\sigma_2)=2\cdot\frac{1}{4}+0\cdot\frac{3}{4}=\frac{1}{2}\qquad\text{and}\qquad V_1(D,\sigma_2)=3\cdot\frac{1}{4}+1\cdot\frac{3}{4}=\frac{3}{2}.
$$
Of course, these numbers are equal to the ones that are usually denoted by $g_1(C,\sigma_2)$ and $g_2(D,\sigma_2)$, respectively. Now, to compute the value usually denoted by $g_1(\sigma_1,\sigma_2)$, one first constructs one more prospect using the measure $\sigma_1$, that is
$$
p^{(\sigma_1,\sigma_2)}=\left(\frac{1}{2},\frac{1}{8}; \frac{3}{2}, \frac{7}{8}\right),
$$
and finally, again under expected utility theory, one computes the value of this prospect, obtaining the well known value $g_1(\sigma_1,\sigma_2)$.
}
\end{example}

We want to replace the classical values $g_i(\sigma_i,\sigma_{-i})$ with new values $V_i(\sigma_i,\sigma_{-i})$, obtained replacing expected utility theory with cumulative prospect theory. From the example, it is clear that, to compute $V_i(\sigma_i,\sigma_{-i})$ in cumulative prospect theory, we only need to compute first $V_i(s_i,\sigma_{-i})$, for all $s_i\in\text{Pl}_i$, using cumulative prospect theory on the prospects $p^{(s_i,\sigma_{-i})}$, and then compute $V_i(\sigma_1,\sigma_2)$ using cumulative prospect theory on the prospect $p^{(\sigma_i,\sigma_{-i})}$. To make this idea formal, recall that in cumulative prospect theory the outcomes of a prospect are supposed to be ordered in increasing way. It is then useful to associate to each prospect $p=(x_1,p_1;\ldots;x_n,p_n)$, with distinct outcomes\footnote{If this prospect does not contain the zero-payoff, we add it with probability zero.}  $x_i\in\mathbb R$, a permutation $\rho(p)$ that is just the permutation of the $x_i$'s such that $\rho(p)(x_i)<\rho(p)(x_{i+1})$, for all $i$. Now for all $(s_i,s_{-i})\in\text{Pl}_i\times\text{Pl}_{-i}$, we define 
$$
A_i^{(s_i,s_{-i})}:=\left\{s_i'\in\text{Pl}_{-i} : g_i(s_i,s_{-i})=g_i(s_i,s_{-i}')\right\}.
$$
For any fixed $s_i$, the sets $A_i$'s form a partition of $\text{Pl}_{-i}$. Choose a transversal $\mathcal T_{s_i}$ for this partition, that is, $\mathcal T_{s_i}$ is a subset of $\text{Pl}_{-i}$ constructed picking exacty one point for each set $A_i$. Now fix $(\sigma_i,\sigma_{-i})\in\mathcal P(\text{Pl}_i)\times\mathcal P(\text{Pl}_{-i})$ and define the prospect
\begin{align*}
p^{(s_i,\sigma_{-i})}=\left(g_i(s_i,s_{-i}),\sigma_{-i}\left(A_i^{(s_i,s_{-i})}\right)\right),
\end{align*}
where $s_{-i}$ runs over the transversal $\mathcal T_{s_i}$. Of course, this prospect does not depend on the particular transversal we fixed. Now, the outcomes of this prospect might not be ordered in increasing way. Therefore, before applying cumulative prospect theory to compute $V_i\left(s_i,p^{(s_i,\sigma_{-i})}\right)$ we must apply the permutation $\rho\left(p^{(s_i,\sigma_{-i})}\right)$. Consequently, with the notation as in Section \ref{se:cpt}, we obtain
\begin{align*}
V_i\left(s_i,p^{(s_i,\sigma_{-i})}\right)=\sum_{s_{-i}\in\mathcal T_{s_i}}\pi_{s_{-i}}v\left(\rho\left(p^{(s_i,\sigma_{-i})}\right)(g_i(s_i,s_{-i}))\right).
\end{align*}
To construct the second prospect $p^{(\sigma_i,\sigma_{-i})}$, we follow an analogous procedure. Let
$$
B_i^{(s_i,\sigma_{-i})}=\left\{s_i'\in\text{Pl}_i : V_i(s_i,\sigma_{-i})=V_i(s_i',\sigma_{-i})\right\}.
$$ 
The $B_i$'s form a partition of $\text{Pl}_i$. Let $\mathcal T_{\sigma_{-i}}$ be a transversal for this partition. We define the prospect
\begin{align*}
p^{(\sigma_i,\sigma_{-i})}=\left(V_i(s_i,\sigma_{-i}),\sigma_i\left(B_i^{(s_i,\sigma_{-i})}\right)\right),
\end{align*}
where $s_i$ runs over $\mathcal T_{\sigma_{-i}}$. Therefore we obtain
\begin{align*}
V_i(\sigma_i,\sigma_{-i})=\sum_{s_i\in\mathcal T_{\sigma_{-i}}}\pi_{s_i}v\left(\rho\left(p^{(\sigma_i,\sigma_{-i})}\right)\left(\sum_{s_{-i}\in\mathcal T_{s_i}}\pi_{s_{-i}}v\left(\rho\left(p^{(s_i,\sigma_{-i})}\right)\right)(g_i(s_i,s_{-i}))\right)\right).
\end{align*}

One is now tempted to define a Nash equilibrium of a game under cumulative prospect theory as a profile $(\sigma_1,\ldots,\sigma_N)$ of mixed strategies such that for all $i\in P$ and for all $\sigma_i'\in\mathcal P(S_i)$ one has $V_i(\sigma_i,\sigma_{-i})\geq V_i(\sigma_i',\sigma_{-i})$. As mentioned before, unfortunately, there are games without Nash equilibria in this sense. To avoid this problem, we use Crawford's trick to extend the set of Nash equilibria including the so-called equilibria in beliefs. To do that, first we recall the following classical definition.

\begin{definition}
{\rm Let $\mathcal D\subseteq\mathbb R^n$ be a convex set and let $\phi:\mathcal D\to\mathbb R$ be a function. The upper contour set of $\phi$ at $a\in\mathbb R$ is the set
$$
U_\phi(a)=\left\{x\in\mathcal D : \phi(x)\geq a\right\}.
$$
$g$ is called quasiconcave on $\mathcal D$ if $U_\phi(a)$ is a convex set for all $a\in\mathbb R$.}
\end{definition} 

The following definition appeared in \cite{Cr90}, Definition 3. In this definition the word \emph{game} is used to denote a classical finite game in normal form $\mathcal G=\mathcal G(P,S,u)$, where the utility functions are extended to the mixed strategies in a possibly non-linear manner.

\begin{definition}
{\rm The convexified version of a game is obtained from the game by replacing each player's preferences by the quasiconcave preferences whose upper contour sets are the convex hulls of his original upper contour sets, leaving other aspects of the game unchanged.}
\end{definition}

We now define Crawford's equilibria in beliefs through an equivalent condition proved by Crawford himself in \cite{Cr90}, Theorem 1.

\begin{definition}\label{defin:beliefs}
{\rm An equilibrium in beliefs is any Nash equilibrium of the convexified version of the game.}
\end{definition}

Crawford proved in \cite{Cr90}, Observation 1, that a Nash equilibrium is always an equilibrium in beliefs and, in Observation 2, that the set of equilibria in beliefs coincides with the set of Nash equilibria if the players have quasiconcave preferences.

We can now define the cooperative equilibria of a game in explicit form.

\begin{definition}\label{defin:cooperative equilibrium cpt}
{\rm The cooperative equilibria of a game in explicit form $\mathcal G=\mathcal G(P,S,\mathfrak g,g,a,f)$ are obtained applying to the reduced game $\text{Red}(\mathcal G)$ the procedure described in Section \ref{se:cooperative equilibrium eut}, replacing
\begin{itemize}
\item the function $g_i(\sigma)$ with the function $V_i(\sigma)$,
\item the notion of Nash equilibrium with the notion of equilibrium in beliefs,
\item the value function $v_i(p)$ in (\ref{eq:value}) with the one in (15).
\end{itemize}
}
\end{definition}

\begin{theorem}\label{th:existence}
Cooperative equilibria exist for all finite games in explicit form.
\end{theorem}

\begin{proof}
Let $\mathcal G=\mathcal G(P,S,\mathfrak g,g,a,f)$ be a finite game in explicit form. We have already proved in Section \ref{se:playable} that the iterated deletion of strategies leads to a well defined and non-empty subgame $\text{Red}(\mathcal G)$. We shall prove that the construction in Section \ref{se:cooperative equilibrium eut} can be applied to $\text{Red}(\mathcal G)$. 

Fix a coalition structure $p$ and let $\mathcal G_p$ be the game obtained by $\text{Red}(\mathcal G)$ grouping together the players in the same coalition, as in Equation (\ref{eq:grouping}). By Crawford's theorem (see \cite{Cr90}, Theorem 2), the set of equilibria in beliefs of $\mathcal G_p$ is not empty. Indeed, this is just the set of Nash equilibria of the convexified game. Now, since the preferences in cumulative prospect theory are described by a continuous function and since continuity is preserved by passing to the convexified version (see \cite{Ro70}, Theorem 17.2), it follows that the set of equilibria in beliefs of $\mathcal G_p$ is compact. Consequently, the sets $M(p_\alpha,p)$ in Equation (\ref{eq:maximal}) are non-empty and the definition of the induced game $\text{Ind}(\mathcal G,p)$ goes through. Observe that the induced game is not empty, since the value of a prospect is at most as the maximal outcome of the prospect, which is an infimum of values attained by the composed function $v\circ V_i$. Therefore, the set of $\sigma$'s such that $(v\circ V_i)(\sigma)\geq v^{\text{CPT}}_i(\sigma)$ is non-empty. Consequently, the set of mixed strategies of the induced game is a non-empty convex and compact subset of the set of mixed strategies of the original game $\mathcal G$. Since in the convexified version of a game the set of mixed strategies does not change, the convexified version of $\text{Ind}(\mathcal G,p)$ has a non-empty set of Nash equilibria (Indeed, observe that Nash's proof of existence of equilibria goes through also if only distinguished convex and compact subsets of mixed strategies are allowed). Applying Theorem 1 in \cite{Cr90}, it follows that the induced game $\text{Ind}(\mathcal G,p)$ has a non-empty set of equilibria in beliefs. Hence, Definition \ref{defin:exact cooperative equilibrium eut} defines a non-empty notion of equilibrium.

Consequently, Definition \ref{defin:cooperative equilibrium cpt} defines a non-empty notion of equilibrium.
\end{proof}

The following corollary follows straight from the construction.

\begin{corollary}\label{cor:predictive}
The exact cooperative equilibrium of a game $\mathcal G$ does not depend on the fairness functions and on the altruism parameters, if
\begin{enumerate}
\item $\mathcal G$ does not have any super-dominated strategies,
\item for every coalition structure $p$, the game $\mathcal G_p$ has a unique equilibrium in beliefs.
\end{enumerate}
\end{corollary}

\begin{remark}\label{ex:asymmetric matching pennies}
{\rm Also in this case we may define the quantal cooperative equilibrium under cumulative prospect theory and taking into account altruism: agent $i$ plays with probability $e^{v_i^{\text{CPT}}(p)}/\sum_pe^{v_i^{\text{CPT}}(p)}$ a quantal level-k solution of the induced game $\text{Ind}(\text{Red}(\mathcal G),p)$. Such quantal cooperative equilibrium explains deviations from Nash equilibrium that have been observed also in purely competitive games, as the asymmetric matching pennies experimented in \cite{Go-Ho01}, that is, the game with gains:

$$
  \begin{array}{ccc}
     & L & R \\
    U & 320,40 & 40,80 \\
    D & 40,80 & 80,40 \\
  \end{array}
$$

It was reported in \cite{Go-Ho01} that most of \emph{vertical players} played the strategy $U$ and most of the \emph{horizontal players} played the strategy $R$. Observe that the Nash equilibrium for the vertical player is the uniform measure on $\{U,D\}$, since the gains of the horizontal player are the same as in the matching pennies. We believe that this behavior ultimately relies in a mistake of the vertical players due to the illusion of a large gain and this mistake is predicted by the horizontal player. This interpretation is confirmed by the cooperative equilibrium. Indeed, the value of the cooperative coalition is easily seen to be equal to $40$ for both players and, therefore, exact cooperative equilibrium reduces to the Nash equilibrium and quantal cooperative equilibrium reduces to the quantal level-k solution. The latter one performs well in such a situation: if the vertical player makes the mistake to think that the horizontal player is level-0 and then she or he is indefferent between playing $L$ and $R$, then the vertical player would have a strong incentive to play the strategy $U$. At this point, the assumption that the horizontal player is level-2 implies that she or he best responds (up to a small mistake) to the strong deviation towards $U$, which is a strong deviation towards $R$.

}
\end{remark}

\section{Summary, conclusions and open problems}\label{se:conclusions}

Over the last decades it has been realised that all classical solution concepts for one-shot normal form games fail to predict human behavior in several strategic situations. 

The purpose of this paper was to attribute these failures to two basic problems, the use of utility functions and the use of solution concepts that do not take into account human attitude to cooperation. While the former problem could be theoretically overcome replacing utility functions by gain functions and applying cumulative prospect theory, the second problem needs a different analysis of the structure of a game. We founded this new analysis on a seemingly reasonable principle of cooperation.
\begin{itemize}
\item[(C)] Players try to forecast how the game would be played if they formed coalitions and then they play according to their best forecast. 
\end{itemize}

To make this idea formal, it has required some effort. In Section \ref{se:explicit form} we have observed that passing from utility functions to gain functions implies that we must take into account new phenomena, such as altruism and perception of gains. We have formalized these phenomena defining the so-called games in explicit form.
After an example describing informally the main idea, in Section \ref{se:cooperative equilibrium eut} we have formalized the principle of cooperation and we have defined the cooperative equilibrium for games in explicit form without using altruism parameters and cumulative prospect theory. The reason of this choice is that altruism and cumulative prospect theory play an active role only on a limited class of games. Indeed, in Section \ref{se:examples} we have shown that the cooperative equilibrium without altruism and cumulative prospect theory already performs well in a number of relevant games. In Section \ref{se:towards cpt} we have discussed a few examples where cumulative prospect theory starts playing an active role and, after a short introduction to cumulative prospect theory in Section \ref{se:cpt}, we have started to adapt the definition given of cooperative equilibrium given in Definition \ref{defin:exact cooperative equilibrium eut} in order to be applied to every game in explicit form and using cumulative prospect theory.  In Section \ref{se:playable} we have used altruism parameters to delete strategies that are not good for the collectivity. This iterated deletion of strategies leads to define a certain subgame. The study of this subgame (done in Section \ref{se:cooperative equilibrium eut} under expected utility theory and in Section \ref{se:cooperative equilibrium cpt} under cumulative prospect theory) contains all relevant new ideas of the paper, that are, the use of the principle of cooperation and the use of cumulative prospect theory: we have assumed that every players try to forecast how the game would be played if they formed coalition; we have used cumulative prospect theory to define a notion of value of a coalition and then, appealing to some Bernoulli-type principle, we have postulated that agents play according to the coalition with highest value.

As shown in the examples in Section \ref{se:examples}, the theory has many positive consequences: to the best of our knowledge, it is the first theory able to organize the experimental data collected for the Traveler's Dilemma, Prisoner's Dilemma, Nash bargaining problem, Bertrand competition, public good game, ultimatum game, and dictator game. These successful applications and the lack of examples where the cooperative equilibrium fails (qualitatively) to predict human behavior, make us optimistic about this direction of research. Nevertheless, we are perfectly aware that the theory is questionable in several points which deserve more attention in future researches. These points include:
\begin{enumerate}
\item To understand if there are other parameters to be taken into account in the definition of games in explicit form. In particular, there is some evidence that \emph{badness parameters} can play an important role in some situations, one of which is described in the following point. 
\item To understand what happens if the players do not agree in playing according to the same coalition structure. Indeed, the cooperative equilibrium works very well in all examples we have discussed since there is a unique coalition structure $p$ that maximizes the value of all players. What happens if different players have different coalition structures maximizing their own value? Do all players \emph{defect} and play according to the coalition structure generated by the maximizing coalition structures, that is, the coarsest coalition which is finer than all maximizing coalition structures? Or, do the players agree to play the fairest coalition structures? In this latter case, what happens if there are many fairest coalition structures? Do the players play uniformly among them?

The difficulty in understanding this point is due mainly to the lack of relevant examples where this situation happens. In fact, we are aware of only one example, where this situation is \emph{about to happen}. We construct this game taking inspiration by a similar game recently experimented in \cite{We-Ra}. Two players have the same strategy set $S_1=S_2=S=\{0,1,\ldots,30\}$. The gain functions are as follows:
$$
g_1(x,y)=\left\{
  \begin{array}{lll}
    30-x, & \hbox{if $x\geq y$} \\
    0, & \hbox{if $x<y$}\\
  \end{array}
\right. \qquad\text{and}\qquad
g_2(x,y)\equiv30.
$$
Let us compute the cooperative equilibrium of this game. The unique equilibrium of the cooperative coalition structure $p_c=\{1,2\}$ is $(0,0)$, where both players get $30$. Observe that no players have incentive to deviate from this equilibrium and consequently, the values of $p_c$ are
$$
v_1(p_c)=30\qquad\qquad\text{and}\qquad\qquad v_2(p_c)=30.
$$

Now consider the selfish coalition structure $p_s=(\{1\},\{2\})$. The value for the second player is again $v_2(p_s)=30$, whereas this time one gets $v_1(p_s)=15$. Indeed, this is one of the cases where the natural symmetry of the game implies that we can restrict the set $\widetilde M(\{2\},p_s)$ taking its barycenter. In other words, when player 2 plays according to $p_s$, she is indifferent among her choices and so she plays uniformly. Player 1's best reply to player 2's uniform measure is the uniform measure, that gives payoff 15. Since this is a Nash equilibrium, there are no possible deviations and so $v_2(p_s)=15$.

So in this case, the unique cooperative equilibrium is $(0,0)$. In other words, player 2 favors player 1 playing $0$ and player 1 knows that player 2 is going to favor her and so she plays $0$ as well. This seems a very natural solution but: \emph{Do humans really play $(0,0)$?} 

We tried to simulate this game with colleagues and friends and something interesting apparently came out. One friend, asked to play the game in the role of player 2, said: ``\emph{It depends. If player 1 is very rich, I would play $30$ for sure!}''. The most common question we were asked after explaining the game was: ``\emph{Do I know the other player?}''. After asking to imagine an anonymous situation, the most common answer (nine out of ten) was: ``\emph{Why should I hurt a person that I do not know? I would play $0$.}''. One person said: ``\emph{I don't care! I would pick a number randomly}''.

Of course, these cannot be considered as experimental data, but we believe that they represent however a light evidence that \emph{badness parameters} do exist. It is not yet clear to the author how to manage them from a general point of view and we will postpone the theorization to a new paper hopefully helped by more experimental data. However, we can say right now how these parameters would effect the play of this particular game. We guess that the badness parameters $b_{ij}$ are non-negative real numbers, where $b_{ij}=0$ represents the situation where player i is absolutely \emph{good} against player j, that is, player i favors player j whenever possible, and $b_{ij}=\infty$ represents the situation where player $i$ is absolutely bad against player $j$. As said, it is not yet clear which would be the exact mathematical definition and the exact effect of these parameters on a general game, but the idea is that in this particular version of the Ultimatum game, the second player plays according to the parameter $b_{21}$ and player 1 estimates a priori the parameter $b_{21}$ and plays a best reply to the strategy that player 2 would play if her badness parameters were equal to player 2's estimation.

\item The formula used in Equation (\ref{eq:value}) to compute the value of a coalition seems a quite reasonable one and it meets the experimental data quite well, but it is certainly only a first tentative. More thoughts, possibly supported by more experimental data, may help to understand the value of a coalition. The main point is probably:
\begin{itemize}
\item to understand whether the value should be computed taking into account also deviations towards \emph{safe strategies}.
\end{itemize}
Indeed, consider the two-player game with gain matrix

$$
  \begin{array}{ccc}
     & a & b \\
    a & 1,1 & 0,-k \\
    b & -k,0 & 10,10 \\
  \end{array}
$$

The cooperative equilibrium is $(b,b)$ independently on $k$. Is this reasonable or for $k$ large enough players prefer not to risk and play the safe strategy $(a,a)$?

\item The formula used in Equation (\ref{eq:value}) to compute the value of a coalition is questionable on another point. In the definition of the numbers $R_j(p)$, we have considered the first step of the reasoning: if player $j$ decides to abandon the coalition structure $p$, then another player, say $k$, may do the same either to follow selfish interests or because she or he is clever enough to anticipate player $j$ deviation. But, if player $j$ is also clever enough to anticipate player $k$'s deviation, then player $j$ may \emph{deviate from the deviation}, and so forth. We could continue this reasoning and define the risk $R_j(p)$ to be, roughly speaking, the maximal lost that player $j$ incurs when a profile of strategies that can be reached by a sequence of deviations is played. Of course, this definition would come at the price of a major technical difficulty, but it would be theoretically more appealing, since it would allow to construct a bridge from the cooperative equilibrium theory to another well studied behavioral model. We recall that $\tau_{i,J}(p)$ has been called \emph{prior probability}, since, despite being an apparently very precise evaluation of how player $i$ measures the event ``\emph{players in J abandon the coalition structure $p$}'', it is well possible that a specific player $i$, for personal reasons, evaluates this event in a completely different way. In particular, the number $\tau_{i,\emptyset}$ represents the probability that player $i$ assigns to the event that no players abandon the coalition. The types of players that are usually called, in economic literature, \emph{altruistic} (resp. \emph{selfish}) would then correspond to those players $i$ who compute the value of a coalition setting $\tau_{i,\emptyset}=1$ (resp. $\tau_{i,\emptyset}=0$), independently of the prior value of such a probability. The correspondence between selfish players and players who set $\tau_{i,\emptyset}=0$ fails using the formula in (\ref{eq:value}), since this formula with $\tau_{i,\emptyset}=0$ can still predict cooperation, even though in a smaller rate, as, for instance, in the Traveler's dilemma. 

\item The exact computation of the cooperative equilibrium is hard for several reasons. First because it goes through the computation of the equilibria in beliefs of several\footnote{As observed by J. Halpern in a private communication, it is
implausible that an agent would consider all coalitions.  In even 
moderately
large games, there are just too many of them.  She may consider some natural
coalitions (e.g., the coalition of all agents), but only a relatively 
small number.
Of course, a theory characterizing which coalitions would be considered is
not easy to come by.} (sub)games. These equilibria are computationally hard to find \cite{Da-Go-Pa06}. Second, because it uses cumulative prospect theory, that is computationally harder than expected utility theory. On one hand, the method that we have proposed is perfectly algorithmic and therefore it might be helpful to write a computer program to compute the cooperative equilibria and make easier the phase of test them on easy real-life situations. On the other hand, it would be important to investigate some computationally easier variant. Of course, quantal level-k theory can be seen as a computationally easier variant, but this theory has the serious issue that it would not be predictive, in the sense that one has to conduct experiments to estimate the error parameter. One could try to avoid this problem using the level-k theory (i.e., only bounded rationality).
\item Iterated deletion of strategies using altruism functions in Section \ref{se:playable} was certainly quite sketchy and it is likely that future researches will suggest a different procedure. In particular, the definition of unplayable strategies of the second type for player $i$ requires that only one particular player $j$ receives a large loss. It is possible that this condition is not sufficient to convince player $i$ to renounce to her better strategy, in case when the players in $P\setminus\{i,j\}$ receives a large gain. 
\item We have defined altruism functions operationally, meaning that one could theoretically compute them by conducting an experiment on the generalized dictator game. It would be important to find an operational way to define the fairness functions.
\end{enumerate}

Open problems include:
\begin{enumerate}
\item Many experiments with different purposes should be conducted. Indeed, an interesting fact is that cooperative equilibrium makes sometimes completely new predictions. A stream of experiments should be devoted to verify or falsify these predictions. For instance,
\begin{itemize}
\item Apparently, the cooperative equilibrium is the unique solution concept predicting an increasing rate of cooperation in the public good game, as the marginal return approaches $1$. It seems that this prediction has a partial confirmation from experimental data, but, as far as we know, only one experiment has been devoted to report this behavior, that is, \cite{IWW94}. Analogously, it seems that the cooperative equilibrium (under cumulative prospect theory) is the unique solution concept predicting or, at least, justifying a rate of cooperation in the public good game with a large number of players. Also in this case, we are aware of only one experimental study devoted to observe this unexpected behavior, that is, again, \cite{IWW94}.
\item Apparently, the cooperative equilibrium is the unique solution concept predicting a rate of cooperation in the Prisoner's dilemma depending on the particular gains. It seems that this prediction is partially confirmed by experimental data, but only on the repeated Prisoner's dilemma (see \cite{DRFN08} and \cite{Fu-Ra-Dr12}). Experiments with a one-shot parametrized Prisoner's dilemma should be conducted to verify or falsify this prediction.
\end{itemize}
Another stream of experiments should be devoted to answer some theoretical questions. At this first stage of research, we believe that the most important one is:
\begin{itemize}
\item to understand whether the value of a coalition structure should be computed taking into account also deviations towards \emph{safe strategies}.
\end{itemize}
\item Have a better understanding of the relation between Nash equilibria and cooperative equilibria (under expected utility theory) for two-person zero-sum games, when the players have the same perception of gains. Indeed, Nash equilibrium performs quite well for zero-sum games and it is possible that all deviations from Nash equilibrium can be explained only making use of cumulative prospect theory. Therefore, it would be important to understand if the cooperative equilibrium (under expected utility theory and assuming $f_1=f_2$) refines Nash equilibrium, in the sense that the set of exact cooperative equilibria is always a subset of the set of Nash equilibria. In this context, it would also be interesting to start from relevant classes of zero-sum games, as the group games, introduced and studied in \cite{Mo10}, \cite{Ca-Mo12}, \cite{Ca-Sc12}. Of course, also a counter-example would be very important to understand if and where the theory can be modified.
\item As stressed several times, the probability $\tau_{i,J}(p)$ is just a prior probability, in the sense that it is well possible that a particular player $i$ computes this probability in a completely different way. It would be important to understand the factors that may influence the evaluation of this probability. For instance, it is well known that individual-level rate of cooperation depends on family history, age, culture, gender, even university course \cite{Ma-Am81}, religious beliefs \cite{HRZ11}, and time decision \cite{RGN12}. The dream is to incorporate this factors into parameters to use to compute the probability $\tau_i$ at an individual-level. Particularly interesting would also be the study of this probability when players can talk each other or have any sort of contact (e.g., eye-contact). Indeed, these contacts can create phenomena of mental reading (see \cite{Wi-MN-GJ}) that we believe can be explained in terms of evaluation of the probability $\tau$.
\end{enumerate}

\end{document}